\documentclass[12pt, draftclsnofoot, onecolumn]{IEEEtran}

\usepackage{subfigure}
\usepackage{setspace}
\usepackage{fancyhdr}
\usepackage{graphicx}
\usepackage{subfigure}
\usepackage{color}
\usepackage{placeins}
\usepackage{float}
\usepackage{flushend}
\usepackage{tabularx,colortbl}
\usepackage{amsmath}
\usepackage{amsfonts}
\usepackage{amssymb}
\usepackage{amsthm}
\usepackage{bm}
\usepackage[numbers,sort&compress]{natbib}
\usepackage{geometry} 

\begin{document}
%
\title{On Uplink Performance of Multiuser Massive MIMO Relay Network With Limited RF Chains}
\author{\small\IEEEauthorblockN{Jindan~Xu\IEEEauthorrefmark{0},~\emph{Student Member,~IEEE}, Yucheng~Wang\IEEEauthorrefmark{0},~\emph{Student Member,~IEEE}, Wei~Xu\IEEEauthorrefmark{0},~\emph{Senior Member,~IEEE}, Shi~Jin\IEEEauthorrefmark{0},~\emph{Senior Member,~IEEE}, Hong~Shen\IEEEauthorrefmark{0},~\emph{Member,~IEEE}, and Xiaohu~You\IEEEauthorrefmark{0},~\emph{Fellow,~IEEE}}\\
\IEEEauthorblockA{
National Mobile Communications Research Laboratory, Southeast University, Nanjing 210096, China\\
Email: \{jdxu, yc.wang, wxu, jinshi, shhseu, xhyu\}@seu.edu.cn}
\vspace{-0.8cm}
\thanks{
Part of this work was presented at the WCSP 2017 in Nanjing, China \cite{wang2017on}.
}}


\maketitle
\thispagestyle{fancy}
\renewcommand{\headrulewidth}{0pt}
\pagestyle{fancy}
\cfoot{}
\rhead{\thepage}
\newtheorem{mylemma}{Lemma}
\newtheorem{mytheorem}{Theorem}
\newtheorem{mypro}{Proposition}
\newtheorem{mycor}{Corollary}

\begin{abstract}
This paper considers a multiuser massive multiple-input multiple-output uplink with the help of an analog amplify-and-forward relay. The base station equips a large array of $N_d$ antennas but is supported by a far smaller number of radio-frequency chains. By first deriving new results for a cascaded phase-aligned two-hop channel, we obtain a tight bound for the ergodic rate in closed form for both perfect and quantized channel phase information. The rate is characterized as a function of a scaled equivalent signal-to-noise ratio of the two-hop channel. It implies that the source and relay powers can be respectively scaled down as $1/N_d^a$ and $1/N_d^{1-a}~ (0\!\leq\!a\!\leq\!1)$ for an asymptotically unchanged sum rate. Then for the rate maximization, the problem of power allocation is optimized with closed-form solutions. Simulation results verified the observations of our derived results.
\end{abstract}



%
\IEEEpeerreviewmaketitle

\section{Introduction}

Multiuser multiple-input multiple-output (MU-MIMO) refers to a system in which a base station (BS) exploits multiple antennas to simultaneously serve many terminals \cite{caire2003on}-\cite{ZZhang2015MIMO}.
Non-orthogonal multiple access (NOMA) holds great promise in carrying massive connectivity in MU-MIMO systems \cite{QYu2019MUMIMO}.
In \cite{ZShi2019AI}, an artificial intelligence (AI) based cooperative spectrum sensing framework was studied for NOMA to improve the spectral efficiency.
In recent years, MU-MIMO has attracted significant interest especially for large-scale antenna arrays, namely massive MIMO \cite{marzetta2010noncooperative}-\cite{xie2016a}. In massive MIMO, small-scale fading is averaged out and the transmit power of each antenna can be aggressively scaled down, leading to significantly improved spectral and energy efficiencies \cite{J_Xu_1}, \cite{ngo2013energy}.


To enhance the coverage of a MU-MIMO, relay has been introduced for the scenarios where direct link between source and destination is weak due to heavy pathloss and shadowing. In \cite{jin2010ergodic}, the capacity of a MU-MIMO relay system was studied, while the degree of freedom of the system was analyzed in \cite{tian2014degree}. Specific transmission designs for the MIMO relay was optimized in \cite{gao2009optical}. Further incorporating the idea of massive MIMO, studies \cite{ngo2014multipair} and \cite{Suraweera2013multi} analyzed the performance of a massive MIMO relay system serving multiple users.
Then in \cite{jin2015ergodic}, the analysis was further explored for a multi-pair two-way amplify and forward (AF) relay massive MIMO system.
For multi-hop communication systems using multiple frequency bands, an algorithm was proposed in \cite{ZFadlullahg2019Multihop} for the relay to choose the optimal modulation method and coding rate.

Precoding is one of the key techniques for achieving the performance in MU-MIMO downlink \cite{jose2011pilot}. Nonlinear precoding such as dirty paper coding (DPC) is known to be capacity-achieving \cite{jindal2005dirty}. However, it is rather complex to implement in practice even in a conventional small-scale MIMO setup. Alternatively, linear precoders such as zero-forcing (ZF) precoding can be adopted to asymptotically approach the theoretical benchmark performance of the optimal nonlinear precoding \cite{rusek2013scaling}.
In the aforementioned works, conventional fully digital signal processing techniques are adopted where each antenna requires a dedicated radio-frequency (RF) chain.
This is in general high-cost especially for the massive MIMO with a large number of antennas. To reduce the hardware and power consumptions, we may resort to constraining the number of RF chains, resulting in a hybrid transceiver architecture \cite{alkhateeb2014channel}, \cite{yu2016hybrid}.

A hybrid precoding scheme, consisting of a digital precoder in baseband and an analog precoder equipped with phase shifters in RF \cite{liang2014low-complexity}, was proven to approach the benchmark performance achieved by conventional fully digital precoding techniques.
In \cite{dai2015near}, a near-optimal iterative hybrid precoding scheme was proposed based on a low-cost sub-array structure.
Then, a successive interference cancelation (SIC)-based hybrid precoder was further studied in \cite{gao2016energy} with reduced computational complexity.
A deep-learning-enabled hybrid precoder was proposed in \cite{HHuang2019learning} for massive MIMO framework.
In \cite{fozooni2016massive}, spectral efficiency was characterized for a multi-pair relay network with a hybrid transceiver. This analysis was then extended in \cite{xu2017spectral} for a multi-pair massive MIMO two-way relay network.

In most of the existing studies on massive MIMO networks using hybrid transceivers, i.e., \cite{liang2014low-complexity}-\cite{xu2017spectral}, the analog processing matrix was designed by extracting the phases of the corresponding single-hop channel.
Few works has investigated the cascaded two-hop relay channels where the hybrid analog processing matrix is designed according to the equivalent cascaded channel, which can be more practical in applications.
In this paper, we study a massive MIMO uplink network assisted with a low-complexity analog relay \cite{borade2007Amplify}-\cite{Park2017Amplify}. Hybrid transceiver architecture is adopted at the BS with limited RF chains. We investigate the performance of the system, especially revealing the effect of limited RF chains on the achievable rate. The derived result quantitatively characterizes the tradeoff between performance and hardware cost. The main contributions of this paper are summarized as follows.
\begin{itemize}
\item The phase of each element of the hybrid analog detecting matrix is element-wisely chosen as that of the cascade two-hop channel matrix. This operation generates an equivalent random channel matrix that does not follow any typical multi-variate distribution as we know. We derive new results on statistics of the phase-aligned cascaded two-hop Gaussian channel matrices, which is shown essential in conducting the performance analysis of the massive relay network with hybrid processing.
\item A tight bound for the achievable rate in massive MIMO relay uplink is obtained in closed form. Further for low signal-to-noise ratios (SNRs) where energy efficiency matters, the ergodic rate of the $k$th user is asymptotically expressed as
    $R_k^{low} = \frac{1}{2}\log_2\left(1 + \frac{\pi}{4} N_d\chi_r \chi_d\right)$
    where $N_d$ is the number of antennas at BS while $\chi_r$ and $\chi_d$ are equivalent SNRs at the relay and BS, respectively. Apparently, there is a decaying factor of $\frac{\pi}{4}$ on the ergodic rate compared to the fully digital scheme.
\item Power scaling laws are obtained to reveal the ability of simultaneously reducing the power consumption of the users and relay as $N_d$ tends to infinity. The asymptotic rate remains constant if the transmit powers of users and relay scale down by $1/N_d^a$ and $1/N_d^{1-a}$ ($0\leq a\leq1$), respectively. Comparison with the full-RF-chain structure verifies that the hybrid detection performs rather close to the fully digital ZF detection in massive MIMO.
\item The optimal power allocation (PA) problem for sum rate maximization of the system is rather intractable due to the nonlinear relationship between the sum rate and PA factors. We impose an auxiliary constraint to the original optimization problem, and transform it into an equivalent one which allows the optimal PA factors to be obtained in closed form through Karush-Kuhn-Tucker (KKT) analysis.
\end{itemize}

The rest of the paper is organized as follows. Section II introduces the system model.
New preliminaries are derived in Section III to assist performance analysis.
In Section IV, asymptotical achievable rate is derived and power saving scenarios are elaborated. Section V deals with the optimal PA problem for sum rate maximization. Numerical results and conclusions are presented in Section VI and VII, respectively.

\emph{Notations}: $(\cdot)^H$ and $(\cdot)^T$ represent the Hermitian and transpose of a matrix, respectively. $(\cdot)^{-1}$ and Tr$(\cdot)$ represent the inverse and trace of a square matrix, respectively. $\mathbb{E}\left\{\right\}$ takes expectation. $\mathbf{I}_K$ denotes the identity matrix of size $K\times K$ while ${\rm diag}(a_1, \cdots, a_n)$ returns a diagonal matrix containing $\{a_1, \cdots, a_n\}$ on the diagonal. $[\cdot]_{ij}$ denotes the $(i,j)$th element of a matrix. $|\cdot|$ and $\|\cdot\|$ take the norm of a complex number and a vector, respectively.
$\angle$ returns the phase of a complex value.
$\mathcal{U}[a,b]$ is the uniform distribution between $a$ and $b$.

\section{System Model}
\begin{figure*}
\centering
\includegraphics[width=6.4in]{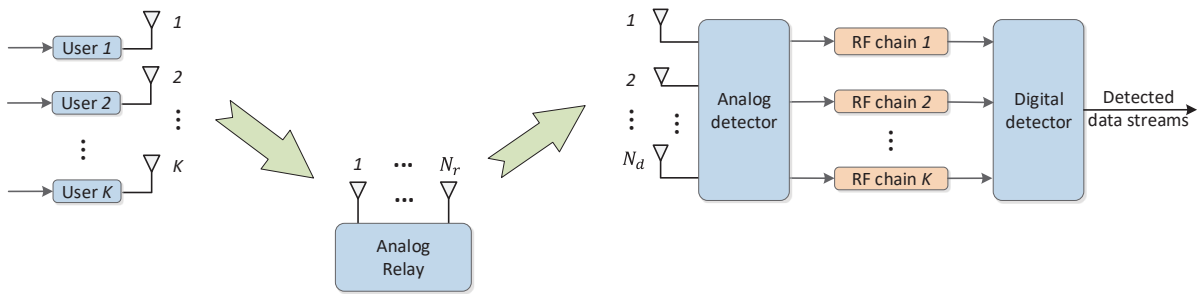}
\caption{Block diagram of a multiuser massive MIMO relaying network with limited RF chains.}
\label{Fig:block}
\end{figure*}

In this paper, we investigate a multiuser massive MIMO relaying uplink where $K$ active users transmit signals to an $N_d$-antenna BS with the help of an $N_r$-antenna analog AF relay, as illustrated in Fig. \ref{Fig:block}. Assume that the BS is equipped with a massive antenna array but driven by a far smaller number, $N_{\rm RF}$, of RF chains.

The channel between $K$ users and the relay can be represented as $\mathbf{H}=[\mathbf{h}_1~\mathbf{h}_2\cdots\mathbf{h}_K]$, where $\mathbf{h}_j\sim \mathcal{CN}(\mathbf{0}_{N_r},\xi_j\mathbf{I}_{N_r})$ is the channel between the $j$th user and the relay and $\xi_j$ denotes the path loss.
The channel between the relay and BS is denoted by $\mathbf{G}\in\mathbb{C}^{N_d \times N_r}$ whose entries are independent and identically distributed (i.i.d.) as $\mathcal{CN}(0,\eta)$. Note that the number of users, $K$, can be generally an arbitrary value while $N_{\rm RF}$ is fixed in a certain application. If $K>N_{\rm RF}$, the system usually schedules a subset with $N_{\rm RF}$ users for simultaneous transmission because the maximum number of independent data streams that can be supported is in theory $N_{\rm RF}$ \cite{Sohrabi2016hybrid}. If $K \leq N_{\rm RF}$, a direct but effective way is to choose only $K$ of the $N_{\rm RF}$ chains for serving the users.
Upon these considerations, we therefore conduct the following discussion and analysis by considering $K = N_{\rm RF}$ without loss of generality.
Then in the following, we describe transmission model in Fig.~\ref{Fig:block} in three steps.

1) $K$ users simultaneously transmit signals to the relay. The receive signal at the relay is
\begin{equation}\label{eq:relay_rec}
\mathbf{y}_r=\sqrt{P_u}\mathbf{H}\mathbf{x}_u + \mathbf{n}_r,
\end{equation}
where $P_u$ is the transmit power of each user, $\mathbf{x}_u = [x_1,\cdots, x_K]^T$ is the transmit signal vector satisfying $\mathbb{E}\{\mathbf{x}_u\mathbf{x}_u^H\}=\mathbf{I}_K$, and $\mathbf{n}_r$ is the Gaussian noise vector satisfying $\mathbb{E}\{\mathbf{n}_r\mathbf{n}_r^H\}=\sigma_r^2\mathbf{I}_{N_r}$.

2) The relay amplifies and forwards the receive signals to BS. The receive signal at BS is
\begin{equation}\label{eq:des_rec}
\mathbf{y}_d=\alpha\mathbf{G}\mathbf{y}_r + \mathbf{n}_d,
\end{equation}
where $\alpha$ is an amplification factor to guarantee the power constrain at the relay, and  $\mathbf{n}_d$ is the Gaussian noise vector satisfying $\mathbb{E}\{\mathbf{n}_d\mathbf{n}_d^H\}=\sigma_d^2\mathbf{I}_{N_d}$. Denoting by $P_r$ the transmit power at the relay, the amplification factor is
\begin{equation}\label{eq:alpha}
\alpha=\sqrt{\frac{P_r}{P_u \text{Tr}(\mathbf{H}^H\mathbf{H})+\sigma_r^2 N_r}}.
\end{equation}

3) The hybrid detection at BS can be composed of an analog RF detector, $\mathbf{W}_a \in \mathbb{C}^{K\times N_d}$, and a subsequent digital baseband detector, $\mathbf{W}_d \in \mathbb{C}^{K\times K}$. Specifically, $\mathbf{W}_a$ exploits phase shifts to adjust the phases of receive signals while $\mathbf{W}_d$ makes adjustments to both signal amplitudes and phases. After the hybrid detection, we have
\begin{equation}\label{eq:est_transmit}
\hat{\mathbf{x}}=\mathbf{W}_d\mathbf{W}_a\mathbf{y}_d.
\end{equation}
Note that there have already existed amounts of design methods for the hybrid processing matrices. We follow a tractable and effective design philosophy, like in \cite{liang2014low-complexity}, where the analog detector is designed based on the cascade channel from users upto the BS, i.e., $\mathbf{GH}$ as a single entity, for several considerations. Firstly, estimating $\mathbf{G}$ and $\mathbf{H}$ is resource consuming and challenging especially at the relay. The channel, $\mathbf{G}\in\mathbb{C}^{N_d\times N_r}$, is even harder to obtain since it is a matrix with a very large number of elements in the massive MIMO. Secondly, the setup with an analog relay as in \cite{borade2007Amplify}-\cite{Park2017Amplify} is more implementable, especially for the massive MIMO relay network. We do not require individual estimates of $\mathbf{G}$ and $\mathbf{H}$ separately, or their statistics. By treating the cascade channel matrix $\mathbf{GH}$ as a single entity, we can apply an uplink channel training for equivalent channel estimation.
With this design, we choose $\mathbf{W}_a$ by extracting the phases of $(\mathbf{GH})^H$, i.e.,
\begin{equation}\label{eq:Fa}
[\mathbf{W}_a]_{ij}=\frac{1}{\sqrt{N_d}}e^{j\phi_{ij}},
\end{equation}
where $\phi_{ij}=\angle [(\mathbf{GH})^H]_{ij}$. While for digital detection, $\mathbf{W}_d$ is designed as a commonly used ZF detector according to the equivalent channel $\mathbf{W}_a\mathbf{GH}$. It follows
\begin{equation}\label{eq:dig_detector}
\mathbf{W}_d=\left(\mathbf{W}_a\mathbf{GH}\right)^{-1}.
\end{equation}
By substituting \eqref{eq:relay_rec}, \eqref{eq:des_rec} and \eqref{eq:dig_detector} into \eqref{eq:est_transmit}, we can rewrite the detected signal as
\begin{equation}\label{eq:est_signal}
\hat{\mathbf{x}}=\alpha\sqrt{P_u}\mathbf{x}_u + \alpha\mathbf{W}_d\mathbf{W}_a\mathbf{G}\mathbf{n}_r + \mathbf{W}_d\mathbf{W}_a\mathbf{n}_d.
\end{equation}
From \eqref{eq:est_signal}, without loss of generality, the received signal-to-interference-plus-noise ratio (SINR) of user $k$ is
\begin{align}\label{eq:SINR1}
\gamma_k
\!=\!\frac{P_u}{\sigma_r^2 [\!\mathbf{W}_d\mathbf{W}_a\mathbf{G}\mathbf{G}^H\mathbf{W}_a^H\mathbf{W}_d^H\!]_{kk} \!+\! \frac{\sigma_d^2}{\alpha^2} [\!\mathbf{W}_d\mathbf{W}_a\mathbf{W}_a^H\mathbf{W}_d^H\!]_{kk}}.
\end{align}
Then, the ergodic rate of the $k$th user can be expressed as
\begin{equation}\label{eq:rate1}
\overline{R}_k=\frac{1}{2}\mathbb{E}\left\{\log_2(1+\gamma_k)\right\}.
\end{equation}

\section{New Preliminaries}

Due to the use of hybrid detection, it is necessary to derive the properties of $\mathbf{W}_a\mathbf{G}$ and $\mathbf{W}_a$ in \eqref{eq:est_signal} for analyzing $\overline{R}_k$.
The main difficulty relies on the dependence of the phases of $\mathbf{W}_a$ and $\mathbf{G}$ according to the design of $\mathbf{W}_a$ in \eqref{eq:Fa}.
The following theorem and two propositions are new essential results which facilitate our following analysis.

\noindent
\begin{mytheorem}\label{mytheorem1}
For independent $N_r\times 1$ random vectors $\mathbf{g}_i$, $\mathbf{g}_j$ and $\mathbf{h}_k$, where $\mathbf{g}_i$ and $\mathbf{g}_j$ are identically distributed as $\mathcal{CN}(\mathbf{0}, \eta\mathbf{I}_{N_r})$, we have
\begin{equation}\label{eq:E gigj}
\mathbb{E}\left\{\mathbf{g}_i^H\mathbf{g}_j e^{j(\phi_1-\phi_2)}\right\} = \begin{cases}
\frac{\pi\eta}{4}  &i\neq j\\
\eta N_r           &i=j,
\end{cases}
\end{equation}
where $\phi_1=\angle \mathbf{g}_j^H\mathbf{h}_k$ and $\phi_2=\angle \mathbf{g}_i^H\mathbf{h}_k$.
\end{mytheorem}
\begin{proof}
The major difficulty comes from the dependence between $\mathbf{g}_i$ and $\phi_k~(k=1,2)$.
Considering $\phi_1=\angle \mathbf{g}_j^H\mathbf{h}_k$ and $\phi_2=\angle \mathbf{g}_i^H\mathbf{h}_k$, the expectation of $\mathbf{g}_i^H\mathbf{g}_j e^{j(\phi_1-\phi_2)}$ in \eqref{eq:E gigj} is taken jointly over different $\mathbf{g}_i$, $\mathbf{g}_j$ and $\mathbf{h}_k$, which implies
\begin{equation}\label{eq:gigjphi1phi2}
\mathbb{E}\!\left\{\!\mathbf{g}_i^H\mathbf{g}_j e^{j(\phi_1-\phi_2)}\!\!\right\} \!\!= \! \mathbb{E}_{\mathbf{h}_k}\!\!\left\{\mathbb{E}_{\mathbf{g}_i,\mathbf{g}_j}\!\left\{\!\mathbf{g}_i^H\mathbf{g}_j e^{j(\phi_1-\phi_2)}\Big|\mathbf{h}_k\!\right\}\!\!\right\}\!\!.
\end{equation}

Firstly we need to calculate the expectation of $\mathbf{g}_i^H\mathbf{g}_j e^{j(\phi_1-\phi_2)}$ for any given $\mathbf{h}_k$ averaging over $\mathbf{g}_i$ and $\mathbf{g}_j$. However, according to the definition of $\phi_1$ and $\phi_2$, $\phi_1-\phi_2$ is also a nonlinear function of $\mathbf{g}_i$ and $\mathbf{g}_j$, which makes the expectation in \eqref{eq:gigjphi1phi2} over $\mathbf{g}_i$ and $\mathbf{g}_j$ hard to evaluate in general.
Fortunately in the following, we are able to prove that the conditional expectation in \eqref{eq:gigjphi1phi2} given $\mathbf{h}_k$ is a value irrelevant to $\mathbf{h}_k$. It implies that we can calculate the expectation conditioned on any realization of $\mathbf{h}_k$, e.g., $\mathbf{h}_k = \mathbf{e}_1 $ where $\mathbf{e}_1$ is an $N_r\times 1$ unit vector whose first element is 1 and others are 0s. By substituting $\mathbf{h}_k = \mathbf{e}_1$ into \eqref{eq:gigjphi1phi2} and dividing the calculation of the expectation into two cases, i.e., $i\neq j$ and $i=j$, we then obtain the analytical result of the expectation.

1) Now, we first prove that the above expectation of $\mathbf{g}_i^H\mathbf{g}_j e^{j(\phi_1-\phi_2)}$ for any given $\mathbf{h}_k$ returns a value that is irrelevant to $\mathbf{h}_k$. Let us focus on the inner conditional expectation taken over $\mathbf{g}_i$ and $\mathbf{g}_j$ in \eqref{eq:gigjphi1phi2}. Firstly we give the singular value decomposition (SVD) of $\mathbf{h}_k$ as
\begin{equation}\label{eq:SVD}
\mathbf{h}_k = \beta \mathbf{U} \mathbf{e}_1,
\end{equation}
where $\beta$ is the singular value of $\mathbf{h}_k$ and $\mathbf{U}$ is an $N_r\times N_r$ unitary matrix. For a certain $\mathbf{h}_k$, there can be many possible $\mathbf{U}$'s. Specifically, we can construct a $\mathbf{U}$ by
\begin{equation}
\mathbf{U} = \left[\frac{\mathbf{h}_k}{\|\mathbf{h}_k\|}, \mathbf{u}_2, \cdots, \mathbf{u}_{N_r}\right],
\end{equation}
where $\{\mathbf{u}_2, \cdots, \mathbf{u}_{N_r}\}$ are a spanned orthogonal basis which makes $\mathbf{U}$ unitary, and $\beta = \|\mathbf{h}_k\|$.
Using \eqref{eq:SVD}, the expectation of $\mathbf{g}_i^H\mathbf{g}_j e^{j(\phi_1-\phi_2)}$ conditioned on $\mathbf{h}_k$ follows
\begin{align}\label{eq:given hk1}
\nonumber
&\mathbb{E}_{\mathbf{g}_i,\mathbf{g}_j}\left\{\mathbf{g}_i^H\mathbf{g}_j e^{j(\phi_1-\phi_2)}\Big|\mathbf{h}_k\right\}\\\nonumber
\overset{(a)}{=}&\mathbb{E}_{\mathbf{g}_i,\mathbf{g}_j}\left\{\mathbf{g}_i^H\mathbf{g}_j e^{j(\phi_1-\phi_2)}\Big|\mathbf{h}_k=\beta\mathbf{U}\mathbf{e}_1\right\}\\\nonumber
\overset{(b)}{=}&\mathbb{E}_{\mathbf{g}_i,\mathbf{g}_j}\left\{\mathbf{g}_i^H\mathbf{g}_j e^{j(\phi_1-\phi_2)}\Big|\mathbf{h}_k=\mathbf{U}\mathbf{e}_1\right\}\\\nonumber
\overset{(c)}{=}&\mathbb{E}_{\mathbf{g}_i,\mathbf{g}_j}\left\{\left(\mathbf{U}^H\mathbf{g}_i\right)^H\mathbf{U}^H\mathbf{g}_j e^{j(\phi_1'-\phi_2')}\Big|\mathbf{h}'_k=\mathbf{e}_1\right\}\\\nonumber
\overset{(d)}{=}&\mathbb{E}_{\mathbf{g}'_i,\mathbf{g}'_j}\left\{\mathbf{g'}_i^{H}\mathbf{g}'_j e^{j(\phi_1'-\phi_2')}\Big|\mathbf{h}'_k=\mathbf{e}_1\right\}\\
\overset{(e)}{=}&\mathbb{E}_{\mathbf{g}_i,\mathbf{g}_j}\left\{\mathbf{g}_i^H\mathbf{g}_j e^{j(\phi_1-\phi_2)}\Big|\mathbf{h}_k=\mathbf{e}_1\right\},
\end{align}
where $(a)$ applies \eqref{eq:SVD}, $(b)$ follows from the fact that $\phi_1$ and $\phi_2$ are irrelevant to the scaling factor $\beta$,
$(c)$ comes from the fact that the left multiplication of unitary $\mathbf{U}^H$ does not change the value of $\mathbf{g}_i^H\mathbf{g}_j$,
i.e., $\left(\mathbf{U}^H\mathbf{g}_i\right)^H\mathbf{U}^H\mathbf{g}_j= \mathbf{g}_i^H \mathbf{U} \mathbf{U}^H \mathbf{g}_j=\mathbf{g}_i^H\mathbf{g}_j$ for unitary $\mathbf{U}^H$.
Then, let us define a new vector as $\mathbf{h}'_k \triangleq \mathbf{U}^H \mathbf{h}_k$ and we have $\mathbf{h}'_k=\mathbf{e}_1$ because of unitary $\mathbf{U}^H$ and the condition in (b), i.e., $\mathbf{h}_k=\mathbf{U}\mathbf{e}_1$.
Now we have
\begin{align}
\nonumber
\phi_1' \triangleq \angle \left(\mathbf{U}^H\mathbf{g}_j\right)^H\mathbf{h}'_k
= \angle \mathbf{g}_j^H\mathbf{U}\mathbf{e}_1=\angle \mathbf{g}_j^H\mathbf{h}_k=\phi_1,\\
\phi_2' \triangleq \angle \left(\mathbf{U}^H\mathbf{g}_i\right)^H\mathbf{h}'_k
= \angle \mathbf{g}_i^H\mathbf{U}\mathbf{e}_1=\angle \mathbf{g}_i^H\mathbf{h}_k=\phi_2,
\label{phi22}
\end{align}
which implies $\phi_1'-\phi_2'=\phi_1-\phi_2$.
These above steps then establish the equality from (b) to (c).
In step (d), we simply use definitions of vectors that $\mathbf{g}'_i \triangleq \mathbf{U}^H\mathbf{g}_i$ and $\mathbf{g}'_j \triangleq \mathbf{U}^H\mathbf{g}_j$.
According to \eqref{phi22}, we have
$\phi_1'=\angle \mathbf{g'}_j^H\mathbf{h}'_k$
and
$\phi_2'=\angle \mathbf{g'}_i^H\mathbf{h}'_k$.
Finally, we obtain (e) due to the fact that $\mathbf{g}'_i$ and $\mathbf{g}'_j$ have the same distribution with $\mathbf{g}_i$ and $\mathbf{g}_j$, respectively
\cite[Th. 3.7.10]{Gallager2013}.
From \eqref{eq:given hk1}, it shows that the conditional expectation for any $\mathbf{h}_k$ is irrelevant to $\mathbf{h}_k$ and the expectation equals to that with $\mathbf{h}_k = \mathbf{e}_1$. Substituting \eqref{eq:given hk1} in \eqref{eq:gigjphi1phi2}, we have
\begin{equation}\label{eq:E gigje1}
\mathbb{E}\left\{\mathbf{g}_i^H\mathbf{g}_j e^{j(\phi_1-\phi_2)}\right\} = \mathbb{E}_{\mathbf{g}_i,\mathbf{g}_j}\left\{\mathbf{g}_i^H\mathbf{g}_j e^{j(\phi_1-\phi_2)}\Big|\mathbf{e}_1\right\}.
\end{equation}

2) In the sequel, we deduce the analytical expression of the expectation in \eqref{eq:E gigje1} by separating calculations for the two cases, i.e., $i\neq j$ and $i=j$.

\noindent
(a) For any $i\neq j$, we have
\begin{align}\label{eq:gi gj e1}
\mathbb{E}_{\mathbf{g}_i,\mathbf{g}_j}\!\left\{\!\mathbf{g}_i^H\mathbf{g}_j e^{j(\phi_1-\phi_2)}\Big|\mathbf{e}_1\!\right\}
\!\overset{(a)}{=}\mathbb{E}_{\mathbf{g}_i}\!\left\{\!\mathbf{g}_i^H e^{-j\varphi_{i1}}\!\right\}\mathbb{E}_{\mathbf{g}_j}\left\{\!\mathbf{g}_j e^{j\varphi_{j1}}\!\right\}\!,
\end{align}
where $\varphi_{i1}$ and $\varphi_{j1}$ are the phases of the first element of $\mathbf{g}_i^H$ and $\mathbf{g}_j^H$, respectively, and $(a)$ applies the independence of $\mathbf{g}_i$ and $\mathbf{g}_j$. Concerning the term $\mathbb{E}_{\mathbf{g}_i}\left\{\mathbf{g}_i^H e^{-j\varphi_{i1}}\right\}$, we have
\begin{equation}\label{eq:gi1}
\mathbb{E}_{g_{i1}}\left\{g_{i1} e^{-j\varphi_{i1}}\right\} = \mathbb{E}\left\{|g_{i1}|\right\} \overset{(a)}{=}\frac{\sqrt{\pi\eta}}{2},
\end{equation}
where $g_{i1}$ represents the first entry of $\mathbf{g}_i^H$ and $(a)$ comes from the fact that the magnitude of $g_{i1}$ follows the Rayleigh distribution with mean $\frac{\sqrt{\pi\eta}}{2}$, since $g_{i1}$ is distributed as $\mathcal{CN}(0, \eta)$.
For the $k$th $(k\neq1)$ element of $\mathbb{E}_{\mathbf{g}_i}\left\{\mathbf{g}_i^H e^{-j\varphi_{i1}}\right\}$, we have
\begin{align}\label{eq:gi1gik}
\mathbb{E}_{g_{i1},g_{ik}}\left\{g_{ik} e^{-j\varphi_{i1}}\right\}
\overset{(a)}{=}&\mathbb{E}\left\{|g_{ik}|\right\}
\mathbb{E}\left\{e^{j\varphi_{ik}}\right\}\mathbb{E}\left\{e^{-j\varphi_{i1}}\right\}
\overset{(b)}{=} 0,
\end{align}
where $g_{ik}$ represents the $k$th element of $\mathbf{g}_i^H$ and $(a)$ is obtained by using the independence of $|g_{ik}|$, $e^{j\varphi_{ik}}$ and $e^{-j\varphi_{i1}}$. As both $g_{i1}$ and $g_{ik}$ follow $\mathcal{CN}(0,\eta)$, their phases follow $\mathcal{U}[0,2\pi)$. Hence $\mathbb{E}\left\{e^{j\varphi_{ik}}\right\}$ =
$\mathbb{E}\left\{e^{-j\varphi_{i1}}\right\}$ = 0 and $(b)$ is obtained.

To conclude, $\mathbb{E}_{\mathbf{g}_i}\left\{\mathbf{g}_i^H e^{-j\varphi_{i1}}\right\}$ can be expressed from \eqref{eq:gi1} and \eqref{eq:gi1gik} as
\begin{equation}\label{eq:gi e1}
\mathbb{E}_{\mathbf{g}_i}\left\{\mathbf{g}_i^H e^{-j\varphi_{i1}}\right\}
= \frac{\sqrt{\pi\eta}}{2}\mathbf{e}_1^H.
\end{equation}
Similarly, we have
\begin{equation}\label{eq:gj e1}
\mathbb{E}_{\mathbf{g}_j}\left\{\mathbf{g}_j e^{-j\varphi_{i1}}\right\}
= \frac{\sqrt{\pi\eta}}{2}\mathbf{e}_1.
\end{equation}
Substituting \eqref{eq:gi e1} and \eqref{eq:gj e1} into \eqref{eq:gi gj e1}, the expectation of $\mathbf{g}_i^H\mathbf{g}_j e^{j(\phi_1-\phi_2)}$ for any $i\neq j$ is
\begin{equation}\label{eq:gi gj1}
\mathbb{E}_{\mathbf{g}_i,\mathbf{g}_j}\left\{\mathbf{g}_i^H\mathbf{g}_j e^{j(\phi_1-\phi_2)}\Big|\mathbf{e}_1\right\}=
\frac{\pi \eta}{4}.
\end{equation}

\noindent
(b) For any $i=j$, we have
\begin{align}\label{eq:gigi1}
\mathbb{E}_{\mathbf{g}_i,\mathbf{g}_j}\left\{\mathbf{g}_i^H\mathbf{g}_j e^{j(\phi_1-\phi_2)}\Big|\mathbf{e}_1\right\}
=&\mathbb{E}\left\{\mathbf{g}_i^H\mathbf{g}_i\right\}
\overset{(a)}{=}2\eta\frac{\Gamma(\frac{N_r}{2}+1)}{\Gamma(\frac{N_r}{2})}
\overset{(b)}{=}\eta N_r,
\end{align}
where $(a)$ is obtained by using the property of Wishart matrix \cite[Th. 3.2.14]{muirhead1982aspects} and $(b)$ is derived by using the property of Gamma function that $\Gamma(x+1)=x\Gamma(x)$ for $x>0$.

Combining \eqref{eq:gi gj1} and \eqref{eq:gigi1}, we have
\begin{equation}\label{eq:vivj1}
\mathbb{E}_{\mathbf{g}_i,\mathbf{g}_j}\left\{\mathbf{g}_i^H\mathbf{g}_j e^{j(\phi_1-\phi_2)}\Big|\mathbf{e}_1\right\} =
\begin{cases}
\frac{\pi\eta}{4}  &i\neq j\\
\eta N_r           &i=j.
\end{cases}
\end{equation}

Finally, the proof completes by substituting \eqref{eq:vivj1} into \eqref{eq:E gigje1}.
\end{proof}

\noindent \emph{Remark}: Theorem~\ref{mytheorem1} does not show the favorable propagation property of massive MIMO as in some existing cases like \cite{ngo2013energy}.
In Theorem~\ref{mytheorem1}, the dependence of $\phi_1$ and $\phi_2$ with $\mathbf{g}_j$ and $\mathbf{g}_i$ makes the common way of applying the law of large numbers (LLN) in massive MIMO not applicable. This dependence among the random variables results in the value of $\mathbb{E}\left\{\mathbf{g}_i^H\mathbf{g}_j e^{j(\phi_1-\phi_2)}\right\}$ in general nonzero for $i\neq j$ which differs from the popular observation in traditional massive MIMO.

Note that the result in Theorem~\ref{mytheorem1} enables us to further obtain the exact expectation result in closed form in Proposition~\ref{mypro1}, which will be shown useful later in the performance analysis.

\noindent
\begin{mypro}\label{mypro1}
Assume two independent channel matrices $\mathbf{H} = [\mathbf{h}_1,\cdots,\mathbf{h}_K]$ where $\mathbf{h}_j\sim \mathcal{CN}(\mathbf{0}_{N_r},\xi_j\mathbf{I}_{N_r})~(j=1,\cdots, K)$ and $\mathbf{G}$ with entries following i.i.d. $\mathcal{CN}(0,\eta)$.
Let $[\mathbf{W}_a]_{ij}=\frac{1}{\sqrt{N_d}}e^{j\phi_{ij}}$ be the, namely phase-aligning, matrix where $\phi_{ij}=\angle\left[(\mathbf{GH})^H\right]_{ij}$. The expectation of the diagonal element of $\mathbf{W}_a\mathbf{G}\mathbf{G}^H\mathbf{W}_a^H$ equals
\begin{equation}\label{eq:appro_fggf}
\mathbb{E}\left\{\left[\mathbf{W}_a\mathbf{G}\mathbf{G}^H\mathbf{W}_a^H\right]_{kk} \right\} = \eta \left(\frac{\pi}{4}N_d + N_r - \frac{\pi}{4}\right),
\end{equation}
where $\eta$ is the variance of the Gaussian entries in $\mathbf{G}$.
\end{mypro}
\begin{proof}
Since the elements of $\mathbf{W}_a$ are extracted from the phases of $(\mathbf{G}\mathbf{H})^H$, $\mathbf{W}_a$ is obviously dependent of $\mathbf{G}$ which makes the expectation difficult to evaluate directly. To tackle this difficulty, we derive the value of $\left[\mathbf{W}_a\mathbf{G}\mathbf{G}^H\mathbf{W}_a^H\right]_{kk}$ by first expanding the term in summation as follows and then element-wisely evaluate the terms with the help of Theorem~\ref{mytheorem1}. From the definition of $\mathbf{W}_a$, we can equivalently write the $(i,j)$th entry of $\mathbf{W}_a$ as
\begin{equation}
[\mathbf{W}_a]_{ij}=\frac{1}{\sqrt{N_d}}\frac{(\mathbf{g}_j^H\mathbf{h}_i)^H}{\big|\mathbf{g}_j^H\mathbf{h}_i\big|},
\end{equation}
where $\mathbf{g}_j^H$ is the $j$th row of $\mathbf{G}$ and $\mathbf{h}_i$ is the $k$th column of $\mathbf{H}$. Now we have
\begin{align}
\left[\mathbf{W}_a\mathbf{G}\mathbf{G}^H\mathbf{W}_a^H\right]_{kk}
=\frac{1}{N_d}\sum_{i=1}^{N_d}\sum_{j=1}^{N_d}\frac{\mathbf{h}_k^H\mathbf{g}_i\mathbf{g}_j^H\mathbf{h}_k}
{\big|\mathbf{g}_i^H\mathbf{h}_k\big|\big|\mathbf{g}_j^H\mathbf{h}_k\big|}\mathbf{g}_i^H\mathbf{g}_j
=\frac{1}{N_d}\sum_{i=1}^{N_d}\sum_{j=1}^{N_d}\mathbf{g}_i^H\mathbf{g}_je^{j(\phi_1-\phi_2)},\label{eq:WaGGWa1}
\end{align}
where $\phi_1=\angle \mathbf{g}_j^H\mathbf{h}_k$ and $\phi_2=\angle \mathbf{g}_i^H\mathbf{h}_k$. Then,
\begin{align}\label{eq:E WGGW1}
\mathbb{E}\left\{\left[\mathbf{W}_a\mathbf{G}\mathbf{G}^H\mathbf{W}_a^H\right]_{kk} \right\}
=\frac{1}{N_d}\sum_{i=1}^{N_d}\sum_{j=1}^{N_d}\mathbb{E}\left\{\mathbf{g}_i^H\mathbf{g}_j e^{j(\phi_1-\phi_2)}\right\}
\overset{(a)}{=}\eta \left(\frac{\pi}{4}N_d + N_r - \frac{\pi}{4}\right),
\end{align}
where $(a)$ applies Theorem~\ref{mytheorem1}. This completes the proof.
\end{proof}

\noindent
\begin{mypro}\label{mypro2}
Let $\mathbf{H}$, $\mathbf{G}$ and $\mathbf{W}_a$ be defined as in the above Proposition~\ref{mypro1}. Then, we have $\phi_{ij}\sim\mathcal{U}[0, 2\pi)$ and
\begin{equation}
\mathbf{W}_a\mathbf{W}_a^H\xrightarrow{a.s.}\mathbf{I}_K,
\end{equation}
where the almost sure convergence, $\xrightarrow{a.s.}$, corresponds to large $N_d$ tending to infinity.
\end{mypro}
\begin{proof}
Since $\phi_{ij}=\angle\left[(\mathbf{GH})^H\right]_{ij}$, we have $\angle [\mathbf{GH}]_{ij}=-\phi_{ji}$ and
\begin{equation}\label{eq:GHij}
[\mathbf{GH}]_{ij}
\overset{(a)}{=}\sum_{k=1}^{N_r}\nu_k e^{j\theta_k}
=\big|[\mathbf{GH}]_{ij}\big|e^{-j\phi_{ji}},
\end{equation}
where $(a)$ is obtained by defining $\nu_k$ and $\theta_k$ as the amplitude and phase of $[\mathbf{G}]_{ik}[\mathbf{H}]_{kj}$, respectively. Since $[\mathbf{G}]_{ik}$ and $[\mathbf{H}]_{kj}$ are independent complex Gaussian random variables (RVs), from \cite[eq.~(17)]{Donoughue2012on}, we obtain that the phase of $[\mathbf{G}]_{ik}[\mathbf{H}]_{kj}$, i.e., $\theta_k$, follows the distribution $\mathcal{U}[0, 2\pi)$ and is independent of $\nu_k$. For $[\mathbf{GH}]_{ij}$, we assume its phase $-\phi_{ji}\in[0,2\pi)$.

Now the uniform distribution of $\phi_{ij}$ can be proved by showing that any two realizations of $\phi_{ij}$ have the same probability of occurrence. From \eqref{eq:GHij}, it show that a given set of values of $\theta_k~(k\in\{1, \cdots, N_r\})$, say $\theta_k^{(0)}$, yields a realization of some $\phi_{ij}$, say $\phi_{ij}^{(0)}$. Meanwhile for an arbitrary fixed value $\theta_0 \in [0, 2\pi)$, a given set of values of $\theta_k = \theta_k^{(0)} + \theta_0$ yields the realization of $\phi_{ij}^{(0)} - \theta_0$. Since $\theta_k$ is uniformly distributed, it is directly known that the probability of occurrence of $\theta_k^{(0)}$ is equal to that of $\theta_k^{(0)}+\theta_0$. It implies that the occurrence of $\phi_{ij}^{(0)}$ is the same as that of  $\phi_{ij}^{(0)}-\theta_0$ for any fixed value of $\theta_0$.

More specifically, for any combination of $\nu_k$ and $\theta_k~(k\in\{1, \cdots, N_r\})$, generating a certain phase $-\phi_{ji}$ as in \eqref{eq:GHij}, we can add $\theta_0$ to each $\theta_k$. It yeilds
\begin{align}
e^{j\theta_0}[\mathbf{GH}]_{ij}
=\sum_{k=1}^{N_r}\nu_k e^{j(\theta_k+\theta_0)}
=\big|[\mathbf{GH}]_{ij}\big|e^{j(\theta_0-\phi_{ji})}
=\big|[\mathbf{GH}]_{ij}\big|e^{j((\theta_0-\phi_{ji}) \bmod 2\pi)}.
\end{align}
In this way, we get another realization of $[\mathbf{GH}]_{ij}$ with phase $((\theta_0-\phi_{ji}) \bmod 2\pi)\in[0,2\pi)$. Due to the uniform distribution of $\theta_k$, one combination of $\theta_k$ and the corresponding combination of $\theta_k+\theta_0$ share the same probability. Hence, all combinations of $\theta_k+\theta_0$ for any $\theta_0\in[0,2\pi)$ have the same probability. Then, all $(\theta_0-\phi_{ji}) \bmod 2\pi$ for any $\theta_0\in[0,2\pi)$ share the same probability. For any other combinations of $\nu_k'$ and $\theta_k'$, we can arrive at the same conclusion. We thus safely obtain that $\angle [\mathbf{GH}]_{ij}$, i.e., $-\phi_{ji}$, is a uniform RV.
Equivalently, we arrive at
\begin{equation}\label{eq:phi}
\phi_{ij}\sim \mathcal{U}[0,2\pi),
\end{equation}
which is the first part of Proposition~\ref{mypro2}.

Subsequently we consider the asymptotic behavior of $\mathbf{W}_a\mathbf{W}_a^H$, whose $p$th diagonal element is
\begin{equation}\label{eq:WaWapp}
\left[\mathbf{W}_a\mathbf{W}_a^H\right]_{pp} = \frac{1}{N_d}\sum_{l=1}^{N_d}e^{j(\phi_{pl}-\phi_{pl})} = 1.
\end{equation}
For the $(p,q)$th non-diagonal element, according to the LLN, we have
\begin{align}\label{eq:WaWapq}
\left[\mathbf{W}_a\mathbf{W}_a^H\right]_{pq}
\xrightarrow{a.s.}\mathbb{E}_{\mathbf{g}_l}\left\{e^{j(\phi_{pl}-\phi_{ql})}\big|\mathbf{g}_l^H\right\}
\overset{(a)}{=}\mathbb{E}\left\{e^{j\phi_{pl}}\right\}\mathbb{E}\left\{e^{-j\phi_{ql}}\right\} =0,
\end{align}
where $(a)$ utilizes the independence of $e^{j\phi_{pl}}$ and $e^{-j\phi_{ql}}$ conditioned on any given $\mathbf{g}_l^H$, and the last equality uses \eqref{eq:phi}. By combining \eqref{eq:WaWapp} and \eqref{eq:WaWapq}, we complete the proof.
\end{proof}

\section{Achievable Rate Analysis}\label{sec:rate}

In this section, asymptotic user rate is derived under the assumption of large antenna arrays.
Power scaling laws are obtained to reveal the tradeoff between power consumption and hardware cost, while keeping a constant user rate.
For notational brevity, we define $\delta \triangleq N_d/N_r$.

\subsection{Asymptotic Rate Analysis}
From \eqref{eq:rate1}, the ergodic achievable sum rate is obtained as
\begin{equation}\label{eq:sumrate1}
\overline{R}=\sum_{k=1}^{K}\overline{R}_k.
\end{equation}
Then we focus on characterizing $\overline{R}_k$. It appears difficult to evaluate the exact value of $\overline{R}_k$ even though we have obtained the preliminary results with respect to the complicated and coupling random matrices, like $\mathbf{W}_a\mathbf{G}$, in \eqref{eq:SINR1}. Here we resort to characterizing a tight performance bound to $\overline{R}_k$ under the massive MIMO setup, as given in the following theorem.

\begin{mytheorem}\label{theorem_rate}
Under the assumption of large antenna arrays, a lower bound for the ergodic user rate is
\begin{align}
R_k=\frac{1}{2}\log_2\left(1 + \frac{\pi \eta \xi_k N_d P_r P_u}{\sigma_r^2 \eta P_r (\pi \delta \!+\! 4) \!+\! 4\sigma_d^2 (P_u \sum_{i=1}^{K}\xi_i \!+\! \sigma_r^2)}\right).
\label{eq:R_k}
\end{align}
\end{mytheorem}
\begin{proof}
By substituting \eqref{eq:SINR1} into \eqref{eq:rate1}, we have
\begin{align}
\nonumber
\overline{R}_k
&\overset{(a)}{=}\frac{1}{2}\mathbb{E}_{x,y}\left\{\log_2\left(1 + \frac{P_u}{\sigma_r^2 x + \sigma_d^2 y}\right)\right\}\\\nonumber
&\overset{(b)}\geq \frac{1}{2}\log_2\left(1 + \frac{P_u}{\sigma_r^2 \mathbb{E}\{x\} + \sigma_d^2 \mathbb{E}\{y\}}\right)\\
&\overset{(c)}{\rightarrow} \frac{1}{2}\log_2\left(1 + \frac{\pi \eta \xi_k N_d P_r P_u}{\sigma_r^2 \eta P_r (\pi \delta \!+\! 4) \!+\! 4\sigma_d^2 (P_u \sum_{i=1}^{K}\xi_i \!+\! \sigma_r^2)}\right)\!,
\end{align}
where the RVs $x$ and $y$ in $(a)$ are defined as
\begin{align}
x \triangleq [\mathbf{W}_d\mathbf{W}_a\mathbf{G}\mathbf{G}^H\mathbf{W}_a^H\mathbf{W}_d^H]_{kk},
y \triangleq \frac{[\mathbf{W}_d\mathbf{W}_a\mathbf{W}_a^H\mathbf{W}_d^H]_{kk}}{\alpha^2},
\end{align}
$(b)$ uses Jensen's inequality, and $(c)$ applies Lemmas~\ref{lemma3}-\ref{lemma4} in Appendix~A, followed by the Continuous Mapping Theorem~\cite{convergence}.
\end{proof}

Note that the rate bound in \eqref{eq:R_k} is tight for the massive MIMO setup. It can alternatively be regarded as an accurate approximation of the exact rate, $\overline{R}_k$, as proved in \cite[Lemma 1]{alkhateeb2014channel}. Moreover, the previous analysis in \cite{jin2010ergodic} showed that the effects of MIMO can still be reflected even though an analog AF relay with scalar $\alpha$ is utilized. This is also evidenced from our analytical results, e.g., through \eqref{eq:R_k} in the massive MIMO relay network. It indicates that the ergodic sum rate logarithmically increases with respect to the number of antennas at the relay, i.e., $N_r$, which implies an array gain of $N_r$.

Let $\chi_r \triangleq \frac{\xi_kP_u}{\sigma_r^2}$ and $\chi_d \triangleq \frac{\eta P_r}{\sigma_d^2}$ denote the equivalent SNRs at the relay and BS, respectively. We can rewrite \eqref{eq:R_k}
equivalently as
\begin{equation}\label{eq:R_k2}
R_k = \frac{1}{2}\log_2\left(1 + \frac{\frac{\pi}{4}N_d\chi_r\chi_d}{A_k + B_k + 1}\right),
\end{equation}
where
\begin{equation}\label{eq:A_k}
A_k = \left(\frac{\pi}{4}\delta + 1\right)\chi_d,
~B_k = \overline{\xi}_k^{-1}\chi_r,
\end{equation}
where $\overline{\xi}_k = \frac{\xi_k}{\sum_{i=1}^K \xi_k}$ in \eqref{eq:A_k} is the normalized large-scale fading from user $k$ to the relay.

From \eqref{eq:R_k2}, it is observed that $R_k$ logarithmically increases with large $N_d$. The term $\chi_r\chi_d$ in \eqref{eq:R_k2} is the product of the equivalent SNRs at the relay and BS. \emph{The coefficients before $\chi_r\chi_d$, $\frac{\pi}{4}$ and $N_d$, respectively represent the effect of limited RF chains and array gains on the achievable rate.}

Moreover, from \eqref{eq:A_k}, $A_k$ is associated with $\chi_d$. It evaluates the influence of the equivalent SNR at BS on the asymptotic rate. \emph{The factor $\frac{\pi}{4}\delta + 1$ is caused by the hybrid detection} and it increases with $\delta$ causing the rate degradation. $B_k$ is associated with $\chi_r$. It represents the impact of the equivalent SNR at relay on the rate. The coefficient $\overline{\xi}_k^{-1}$ comes from the amplification factor of the relay and it increases with $K$ which also contributes to the rate degradation of each data stream.

More specifically, we then investigate the impact of different equivalent SNRs at the relay and BS on the achievable rate. The following three typical cases are discussed.

1) Case 1: \emph{Low SNR Analysis}. For $\chi_d\ll 1$ and $\chi_r\ll 1$, $R_k$ in \eqref{eq:R_k2} can be expressed as follows:
\begin{equation}\label{eq:Rk3}
R_k^{low} = \frac{1}{2}\log_2\left(1 + \frac{\pi}{4} N_d\chi_r \chi_d\right).
\end{equation}

In the low-SNR scheme, the achievable rate is a function of a scaled product of equivalent SNRs at relay and BS. The SINR of each user is proportional to $\chi_r\chi_d$. The scaling factor equals $\frac{\pi}{4}N_d$ where $N_d$ comes from the array gain and $\frac{\pi}{4}$ is due to the effect of limited RF chains.

2) Case 2: \emph{High SNR Analysis}. For $\chi_d\gg 1$ and $\chi_r\gg 1$, $R_k$ in \eqref{eq:R_k2} can approximately be given by
\begin{equation}\label{eq:Rk4}
R_k^{high} = \frac{1}{2}\log_2\left(1 + \frac{\pi}{\left(\pi \delta \!+\! 4\right)\chi_d \!+\! 4 \overline{\xi}_k^{-1}\chi_r}N_d\chi_r \chi_d\right).
\end{equation}

In the high-SNR scheme, the achievable rate is a nonlinear function of $\chi_r$ and $\chi_d$ due to the relatively complex relationship between SINR and $\chi_r$ and $\chi_d$ under the relay network with limited RF chains. Assuming the same path loss from users to relay, we have $\overline{\xi}_k^{-1} = K$.

3) Case 3: Intuitively, when the SNR at BS is far higher than that at relay, the relay system degenerates into a single-hop one. Specifically for $\chi_d\gg\chi_r$ and $\chi_d\gg 1$, we have $A_k\gg B_k$ and $A_k\gg C_k$. $R_k$ in \eqref{eq:R_k2} is approximately given by
\begin{equation}\label{eq:Rk1}
R_k^{(3)} = \frac{1}{2}\log_2\left(1 + \frac{\pi\delta}{\pi \delta + 4}N_r\chi_r\right).
\end{equation}
From \eqref{eq:Rk1}, $R_k$ is rarely affected by $\chi_d$. The achievable rate only depends on the channel parameter from users to relay. Compared to the achievable rate of the single-hop system using pure digital detection as studied in \cite{ngo2013energy}, \emph{the hybrid processing introduces an SINR reduction by a multiplier factor $\frac{\pi\delta}{\pi \delta + 4}$.} As $\delta\rightarrow\infty$, the ergodic rate achieved by hybrid detection approaches to that achieved by fully digital detection.

4) Case 4: On the other hand, when the SNR at BS is far lower than that at relay, the relay system also degenerates into a single-hop one. Specifically for $\chi_r\gg\chi_d$ and $\chi_r\gg 1$, we have $B_k\gg A_k$ and $B_k\gg C_k$. $R_k$ in \eqref{eq:R_k2} is approximately given by
\begin{equation}\label{eq:Rk2}
R_k^{(4)} = \frac{1}{2}\log_2\left(1 + \frac{\pi}{4}\overline{\xi}_kN_d\chi_d\right).
\end{equation}

\subsection{Power Scaling Law}
In this section, we assume $\xi_k=1$ for all $k$, and normalize $\eta=1$. We focus on a non-decreasing achievable rate when the transmit power of users and/or relay is reduced as the number of antennas increases, i.e., $P_u=\frac{E_u}{N_d^a}$ and $P_r=\frac{E_r}{N_d^b}$ $(a,b\geq0)$ for fixed $E_u$ and $E_r$. Let $\gamma_r\triangleq\frac{E_u}{\sigma_r^2}$ and $\gamma_d\triangleq\frac{E_r}{\sigma_d^2}$ in the following analysis. Rate in \eqref{eq:R_k} simplifies to
\begin{align}\label{eq:powerscaling1}
R_k
= \frac{1}{2}\log_2\left(1 \!+\! \frac{\pi \gamma_r \gamma_d}{(\pi \delta + 4)\gamma_d N_d^{a\!-\!1} + 4K\gamma_r N_d^{b-1} + 4N_d^{a+b\!-\!1}}\right).
\end{align}

In order to obtain a non-vanishing rate with increasing $N_d$, we arrive at the condition $a+b \leq 1$.
To reduce the power consumption as much as possible, we consider $a+b = 1$. Then, a power-saving scheme is obtained where $a$ is an arbitrary nonnegative number satisfying $0\leq a\leq 1$, which can be separately treated in three cases as follows.

1)
As $P_u=E_u/N_d^a, P_r=E_r/N_d^{1-a}~(0<a<1)$, the asymptotic achievable rate is
\begin{align}
\label{eq:iiii}
R_k&\rightarrow\frac{1}{2}\log_2\left(1 + \frac{\pi}{4}\gamma_r \gamma_d\right).
\end{align}
Given $N_d\rightarrow\infty$, \eqref{eq:iiii} is obtained by letting $b=1-a$ in \eqref{eq:powerscaling1}.
We can simultaneously scale down the transmit power of users by $1/N_d^a$ and the relay power by $1/N_d^{1-a}$ while keeping a nearly constant rate. Note that the asymptotic achievable rate in \eqref{eq:iiii} is in fact the rate expression at low SNRs with user power of $\frac{E_u}{N_d^a}$ and relay power of $\frac{E_r}{N_d^{1-a}}$. It is equal to the rate of single-input single-output (SISO) system with user transmit power of $E_u$ and relay forwarding power of $E_r$.

2)
As $P_u=E_u/N_d, P_r=E_r$, the asymptotic achievable rate is given as
\begin{align}
\label{eq:iii}
R_k&\rightarrow\frac{1}{2}\log_2\left(1 + \frac{\pi}{(\pi \delta + 4)\gamma_d+4}\gamma_r \gamma_d\right).
\end{align}
Similarly, \eqref{eq:iii} is obtained by letting $a=1$ and $b=0$ in \eqref{eq:powerscaling1} and taking the limit as $N_d\rightarrow\infty$.
In this case, the transmit power of users can be scaled down by $1/N_d$ while the forwarding power of relay is kept fixed.

3)
As $P_u=E_u, P_r=E_r/N_d$, the asymptotic achievable rate is given as
\begin{align}
\label{eq:ii}
R_k&\rightarrow\frac{1}{2}\log_2\left(1 + \frac{\pi}{4K \gamma_r+4}\gamma_r \gamma_d\right).
\end{align}
Assuming $N_d\rightarrow\infty$, \eqref{eq:ii} is obtained by letting $a=0$ and $b=1$ in \eqref{eq:powerscaling1}.
We can scale down the forwarding power of relay by $1/N_d$ while the transmit power of each user is kept fixed.

\subsection{Effects of Phase Quantization}
In the previous analysis, we have assumed perfect channel phase information available at the analog detector. Further in this subsection, the above performance analysis is extended to a more practical consideration with quantized channel phase information available for the analog processing. The phase of each element of $\hat{\mathbf{W}}_a$ is quantized based on the closest Euclidean distance from a codebook via
\begin{equation}
\Psi=\Big\{0, \pm\frac{1}{2^b}2\pi, \cdots, \pm\frac{2^{b-1}-1}{2^b}2\pi, \pi\Big\},
\end{equation}
where $b$ is the number of quantization bits. We then applies the digital ZF detector $\hat{\mathbf{W}}_d$ on the equivalent channel $\hat{\mathbf{W}}_a\mathbf{GH}$. The presence of phase quantization error changes the obtained results in Proposition~\ref{mypro1} and Lemmas~\ref{lemma2}-\ref{lemma4}.
Correspondingly updated results are derived in Appendix~B.
Using the quantized channel phase information for hybrid detection, a lower bound for the achievable user rate is given in the following theorem.

\begin{mytheorem}\label{theorem_rate_Q}
Under the assumption of large antenna arrays and quantized channel phase information, a lower bound for the ergodic user rate is
\begin{align}
&\hat{R}_k=
\frac{1}{2}\log_2\left(\!\!1 \!+\! \frac{\pi \eta \xi_k N_d P_r P_u{\rm sinc}^2\big(\frac{\pi}{2^b}\!\big)}{\sigma_r^2 \eta P_r (\pi \delta {\rm sinc}^2\big(\frac{\pi}{2^b}\!\big)\!+\! 4) \!+\! 4\sigma_d^2 (\!P_u\! \sum_{i=1}^{K}\xi_i \!+\! \sigma_r^2)}\!\right)\!.
\label{eq:R_k_Q}
\end{align}
\end{mytheorem}
\begin{proof}
By substituting \eqref{eq:SINR1} into \eqref{eq:rate1} and replacing $\mathbf{W}_a$ and $\mathbf{W}_d$ with $\hat{\mathbf{W}}_a$ and $\hat{\mathbf{W}}_d$ respectively, we have
\begin{align}
\nonumber
\overline{R}_k
&\overset{(a)}{=}\frac{1}{2}\mathbb{E}_{x,y}\left\{\log_2\left(1 + \frac{P_u}{\sigma_r^2 \hat{x} + \sigma_d^2 \hat{y}}\right)\right\}\\\nonumber
&\overset{(b)}\geq \frac{1}{2}\log_2\left(1 + \frac{P_u}{\sigma_r^2 \mathbb{E}\{\hat{x}\} + \sigma_d^2 \mathbb{E}\{\hat{y}\}}\right)\\
&\overset{(c)}{\rightarrow} \frac{1}{2}\log_2\left(\!1+ \!
\! \frac{\pi \eta \xi_k N_d P_r P_u{\rm sinc}^2\big(\frac{\pi}{2^b}\big)}{\sigma_r^2 \eta P_r (\pi \delta {\rm sinc}^2\big(\frac{\pi}{2^b}\big)\!+\! 4) \!+\! 4\sigma_d^2 (P_u \sum_{i=1}^{K}\xi_i \!+\! \sigma_r^2)}\!\right),
\end{align}
where the RVs $\hat{x}$ and $\hat{y}$ in $(a)$ are defined as
\begin{align}
\hat{x} &\triangleq [\hat{\mathbf{W}}_d\hat{\mathbf{W}}_a\mathbf{G}\mathbf{G}^H\hat{\mathbf{W}}_a^H\hat{\mathbf{W}}_d^H]_{kk},
\hat{y} \triangleq \frac{[\hat{\mathbf{W}}_d\hat{\mathbf{W}}_a\hat{\mathbf{W}}_a^H\hat{\mathbf{W}}_d^H]_{kk}}{\alpha^2},
\end{align}
$(b)$ uses Jensen's inequality, and $(c)$ applies \eqref{eq:lemma2_Q}--\eqref{eq:lemma3_Q} in Appendix~B, followed by applying the Continuous Mapping Theorem \cite{convergence}.
\end{proof}

\section{Power Allocation}
In the previous analysis, we assume the same transmit power for each user which in general is not optimal when it comes to maximizing the sum rate. In this section, we assume $P_T$ as the total transmit power budget of all users, and denote $P_{u,k}$ as the transmit power of the $k$th user, say $P_{u,k}=\mu_k P_T$ where $\mu_k$ satisfies $\sum_{k=1}^K \mu_k =1$ and $\mu_k \geq 0$.
Assume that $P_T,P_r,N_d$, and $\delta$ are fixed, which means users and relay are working with fixed transmit powers and antenna configurations at the relay and BS.
By replacing $P_u$ in \eqref{eq:R_k} with $\mu_k P_T$, we now resort to finding the optimal PA strategy, namely the optimal combination of $\mu_k$, to maximize the sum rate.

It is obvious that the sum rate has the form as
\begin{equation}\label{eq:PA1}
R\triangleq\sum_{k=1}^{K}R_k=\frac{1}{2}\sum_{k=1}^K\log_2\left(1+\frac{l_k\mu_k}{\sum_{i=1}^K m_i\mu_i + n}\right),
\end{equation}
where $l_k,m_i$, and $n$ are constant values defined by
\begin{align}\label{eq:PA2}
\begin{cases}
l_k=\pi\eta N_d P_r P_T \xi_k\\
m_i=4\sigma_d^2 P_T \xi_i\\
n=\sigma_r^2 \eta P_r(\pi\delta+4) + 4\sigma_d^2\sigma_r^2.
\end{cases}
\end{align}
Hence, the PA problem can be written as
\begin{align}\label{eq:PA}
&\max_{\mu_1,\mu_2,\cdots,\mu_K} \sum_{k=1}^K \ln\left(1 + \frac{l_k\mu_k}{\sum_i m_i\mu_i + n}\right),\\
\nonumber
&{\rm s.t.}~~
\sum_{k=1}^K \mu_k=1,~
\mu_k\geq0, k=1,2,\cdots,K.
\end{align}

The optimization problem above is in general complicate and it is infeasible to get the optimal solution directly owing to the nonlinear relationship between the ergodic rate and the PA factors. By introducing an auxiliary constraint as $\sum_i m_i\mu_i = \rho$, we obtain an equivalent problem as
\begin{align}
&~~~~~~~~~~~~~\max_{\mu_1,\cdots,\mu_K,\rho} \sum_{k=1}^K \ln\left(1 + \frac{l_k\mu_k}{\rho + n}\right),
\\\nonumber
&{\rm s.t.}~~
\sum_{k=1}^K \mu_k=1,~\sum_{i=1}^K m_i\mu_i=\rho,~
\mu_k\geq0, k=1,\cdots,K.
\end{align}

To facilitate the analysis, we equivalently transform the equality constraints into inequality ones, which is easily verified that it does not change the optimality of the solution. In the following, KKT analysis is used to solve the optimal PA given by the following theorem.

\begin{mytheorem}\label{theorem3}
Consider the optimization problem
\begin{align}\label{eq:opt pro}
&~~~~~~~~~~~\min_{\mu_1,\cdots,\mu_K,\rho} -\sum_{k=1}^K \ln\left(1 + \frac{l_k\mu_k}{\rho + n}\right),
\\\nonumber
&{\rm s.t.}~~
\sum_{k=1}^K \mu_k\leq1,~
\sum_{i=1}^K m_i\mu_i\leq\rho,~
\mu_k\geq0, k=1,\cdots,K.
\end{align}
The optimal solution is given as
\begin{align}\label{eq:opt sol}
\mu_k=\max\left\{\frac{1}{v+wm_k}-\frac{\rho+n}{l_k},0\right\},
\end{align}
where $v$ and $w$ are chosen by satisfying $\sum_{k=1}^K \mu_k=1$ and $\sum_{i=1}^K m_i\mu_i=\rho$, respectively.
\end{mytheorem}

\begin{proof}
Firstly, we define the Lagrangian of problem \eqref{eq:opt pro} as
\begin{align}
&L(\mu_1,\cdots,\mu_K,u_1,\cdots,u_K,v,w)\nonumber\\
=&\!-\!\sum_k \ln\left(1 \!+\! \frac{l_k\mu_k}{\rho \!+\! n}\right)\!-\!\sum_k u_k\mu_k\!+\!v\left(\sum_k \mu_k\!-\!1\right)
+\!w\left(\sum_i m_i\mu_i\!-\!\rho\right).
\end{align}
Then, the KKT conditions are given by
\begin{align}\label{eq:KKT 1}
\begin{cases}
-\frac{l_k}{\rho+n+l_k\mu_k}-u_k+v+wm_k=0\\
v(\!\sum_k \mu_k\!-\!1\!)=0,~w(\!\sum_i m_i\mu_i\!-\!\rho\!)\!=\!0,~u_k\mu_k\!=\!0\\
\sum_k \mu_k\leq1,~\sum_i m_i\mu_i\leq\rho,~\mu_k\geq0\\
u_k\geq0,~v\geq0,~w\geq0.
\end{cases}
\end{align}
Through further analysis, we have
\begin{align}
&\mu_k=\max\left\{\frac{1}{v+wm_k}-\frac{\rho+n}{l_k},0\right\},\\\nonumber
&{\rm s.t.}~
\begin{cases}
v(\sum_k \mu_k\!-\!1)=0,~w(\sum_i m_i\mu_i\!-\!\rho)=0\\
\sum_k \mu_k\leq1,~\sum_i m_i\mu_i\leq\rho\\
v\geq0~,w\geq0.
\end{cases}
\end{align}

The fact that $v$ and $w$ cannot be zeros simultaneously lead us to consider the following cases:

1)~For $v=0$ and $w>0$, it yields
\begin{equation}\label{eq:casea}
\mu_k\!=\!\max\left\{\!\frac{1}{wm_k}\!-\!\frac{\rho\!+\!n}{l_k},0\!\right\}\!,~{\rm s.t.}~\sum_k \mu_k\!\leq\!1,~\sum_i m_i\mu_i\!=\!\rho.
\end{equation}

Assume that $\mu_k'$ is the optimal solution satisfying \eqref{eq:casea} where $\sum_k \mu_k'<1$ while $\mu_k$ is another solution conforming to \eqref{eq:casea} where $\sum_k \mu_k=1$. Noticeably, there exists some $k$ satisfying $\mu_k'<\mu_k$ and other $k$ satisfying $\mu_k'\leq\mu_k$. Let $R_0'$ and $R_0$ be the values of objective function under $\mu_k'$ and $\mu_k$, respectively. We have $R_0'<R_0$ yielding a conflict. Hence $\mu_k$ corresponding to the maximum sum rate objective must satisfy $\sum_k \mu_k=1$.

2)~For $v>0$ and $w=0$, it yields
\begin{equation}
\mu_k\!=\!\max\left\{\frac{1}{v}\!-\!\frac{\rho\!+\!n}{l_k},0\right\},~{\rm s.t.}~\sum_k \mu_k\!=\!1,~\sum_i m_i\mu_i\!\leq\!\rho.
\end{equation}
Similar as in case 1), the optimal solution is acquired when $\sum_i m_i\mu_i=\rho$.

3)~For $v>0$ and $w>0$, it yields
\begin{equation}
\mu_k\!=\!\max\left\{\frac{1}{v\!+\!wm_k}\!-\!\frac{\rho\!+\!n}{l_k},0\right\},~{\rm s.t.}~\sum_k \mu_k\!=\!1,~\sum_i m_i\mu_i\!=\!\rho.
\end{equation}

Finally, the proof completes by combining the above three cases.
\end{proof}

From Theorem \ref{theorem3}, we can solve $\mu_k$ through, e.g., a three-dimensional search for $\rho$, $v$ and $w$ with the closed-form solution in \eqref{eq:opt sol}. Moreover, the optimal PA patterns under the high-SNR and low-SNR schemes are analyzed in the following corollaries from Theorem \ref{theorem3}.

\begin{mycor}\label{corollary1}
As $P_T$ tends to 0, the optimal solution in Theorem \ref{theorem3} is given as
\begin{align}
\lim\limits_{P_T\rightarrow0}\mu_k=\max\left\{c_1-\frac{c_2}{\xi_k},0\right\},
\end{align}
where $c_1=\frac{1}{v}$ and $c_2=\frac{\rho + \sigma_r^2 \eta P_r (\pi\delta+4) + 4\sigma_d^2\sigma_r^2}{\pi \eta N_d P_r P_T}$.
\end{mycor}
\begin{proof}
Taking the limit of \eqref{eq:opt sol}, we have
\begin{align}
\lim\limits_{P_T\rightarrow0}\mu_k
=&\max\left\{\lim\limits_{P_T\rightarrow0}\left(\frac{1}{v+wm_k}-\frac{\rho+n}{l_k}\right),0\right\}
=\max\left\{\!\frac{1}{v}\!-\!\frac{\rho \!+\! \sigma_r^2 \eta P_r (\pi\delta\!+\!4) \!+\! 4\sigma_d^2\sigma_r^2}{\pi \eta N_d P_r P_T \xi_k},0\!\right\}.
\end{align}
\end{proof}

From Corollary \ref{corollary1}, in the low-SNR regime, $\mu_k$ is monotonically increasing with $\xi_k$. \emph{Users with better channel condition are assigned with more power in order to achieve the maximal sum rate.} The smaller $P_T$ is, the PA factors for difference users differs more significantly from each other.

\begin{mycor}\label{corollary2}
As $P_T$ tends to infinity, the optimal solution in Theorem \ref{theorem3} is given as
\begin{align}\label{eq:c2}
\lim\limits_{P_T\rightarrow\infty}\mu_k=\max\left\{\frac{c_3}{\xi_k},0\right\},
\end{align}
where $c_3=\frac{1}{4 \sigma_d^2 w P_T} - \frac{\rho + \sigma_r^2 \eta P_r (\pi\delta+4) + 4\sigma_d^2\sigma_r^2}{\pi\eta N_d P_r P_T}$.
\end{mycor}

\begin{proof}
The proof is completed by taking the limit of \eqref{eq:opt sol} followed by the similar derivations in Corollary 1.
\end{proof}

From Corollary \ref{corollary2}, in the high-SNR scheme, when $c_3>0$, $\mu_k$ is monotonically decreasing with $\xi_k$. \emph{Users with better channel condition are assigned with less power.} Similarly as Corollary \ref{corollary1}, for extremely large $P_T\rightarrow\infty$, the PA factors for difference users tend to be equal.

It is interesting to see that the optimal PA strategy varies fundamentally for different SNR values by comparing Corollary \ref{corollary1} and Corollary \ref{corollary2}. Specifically, the optimal PA allocates more power to users with better channel condition at low SNR regime, similarly as the water-filling like power control; while the optimal PA strategy behaves in the opposite direction at the high SNR regime which is contrary to the water-filling like strategy.
When $P_T$ tends to 0, the amplification factor tends to infinity so that the impact of the additional noise at BS is negligible. In this case, the relay system degenerates into a single-hop one without the amplification effect of the relay. Hence, the optimal PA conforms to the classic water-filling like solution. When $P_T$ tends to infinity, the amplification factor tends to 0, which makes the multiuser interference introduced by the amplification effect rather pronounced.
Hence, a water-filling like solution could enlarge the interference, which makes the sum rate decrease.
From an SINR perspective, when $P_T$ tends to infinity, we can observe the change in the sum rate of a simplified system with $K = 2$ upon a small variation on $\mu_k$.
The sum rate in \eqref{eq:PA1} changes to
\begin{align}
\label{r1_5_sumrate4}
\tilde{R} \triangleq & \frac{1}{2}\log\Big(1+\frac{c\xi_1(\mu_1+\Delta\mu)}{\xi_1(\mu_1+\Delta\mu)+\xi_2(\mu_2-\Delta\mu)}\Big)
\nonumber
+ \frac{1}{2}\log\Big(1+\frac{c\xi_2(\mu_2-\Delta\mu)}{\xi_1(\mu_1+\Delta\mu)+\xi_2(\mu_2-\Delta\mu)}\Big)\\\nonumber
\triangleq& \frac{1}{2}\log\Big(1+{\rm SINR}_1'\Big) + \frac{1}{2}\log\Big(1+{\rm SINR}_2'\Big)\\
\triangleq& \tilde{R}_{1}+\tilde{R}_{2},
\end{align}
where $c$ is a constant expressed as $c \triangleq \frac{\pi\eta P_rN_d}{4\sigma^2_d}.$ Applying Taylor expression, $\tilde{R}_{k}~(k\in\{1,2\})$ and $\tilde{R}$ in \eqref{r1_5_sumrate4} are written as
\begin{align}
\label{r1_5_R21}
\tilde{R}_{k}
&= \frac{1}{2}\log\Big(1+{\rm SINR}_k\Big)
\!+\!\frac{(-1)^{k+1}}{2} \Big(\frac{1}{1\!+\!{\rm SINR}_k'}\frac{\partial({\rm SINR}_k')}{\partial(\Delta\mu)}\Big)\Big|_{\Delta\mu=0}\Delta\mu \!+\! o(\Delta\mu),
\end{align}
\begin{align}
\label{r1_5_R2}
\tilde{R}
=R+b\Delta\mu+o(\Delta\mu),
\end{align}
where ${\rm SINR}_k\triangleq\frac{c\xi_k\mu_k}{\xi_1\mu_1+\xi_2\mu_2}$ and $b$ is a constant defined as
\begin{align}
\label{r1_5_b}
b=&\frac{c^2\xi_1\xi_2(\mu_1+\mu_2)}{(\xi_1\mu_1+\xi_2\mu_2)(\xi_1\mu_1+\xi_2\mu_2+c\xi_1\mu_1)}
-
\frac{c^2\xi_1\xi_2(\mu_1+\mu_2)}{(\xi_1\mu_1+\xi_2\mu_2)(\xi_1\mu_1+\xi_2\mu_2+c\xi_2\mu_2)}.
\end{align}
Assume that $\xi_1>\xi_2$ and $\mu_1=\mu_2$. If we consider a water-filling like solution, i.e., $\Delta \mu>0$, we have $b<0$ and thus $\tilde{R}<R$, which indicates that our proposed optimal PA outperforms the water-filling like solution.

Then, combining Corollaries \ref{corollary1} and \ref{corollary2} with growing $P_T$, we can increase the fairness of the relay system by allocating a higher proportion of total transmit power to the users with poor channel condition. Hence, there is a tradeoff between the total power consumption and the system fairness when comes to sum rate maximization.

\section{Simulation Results}

In this section, Monte-Carlo simulations are adopted to evaluate the sum rate of the massive MIMO relay network with hybrid detection at the BS. We set normalized $\sigma_r^2 = \sigma_d^2 = 0$~dB and $\eta = 1$.

\begin{figure}
\centering
\includegraphics[width=3.6in]{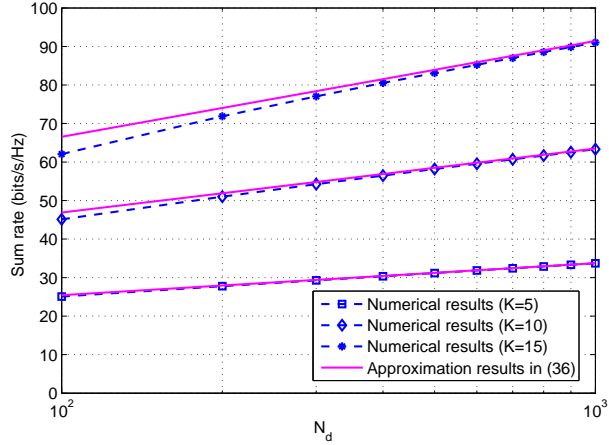}
\caption{Sum rate versus number of antennas with $P_u = P_r$ = 20dB, $\delta=1$, and $\xi_k = 1$ for $k=1, \cdots, K$.}
\label{Fig:sum_rate_comparison_simulation_approximation}
\end{figure}
\begin{figure}
\centering
\includegraphics[width=3.6in]{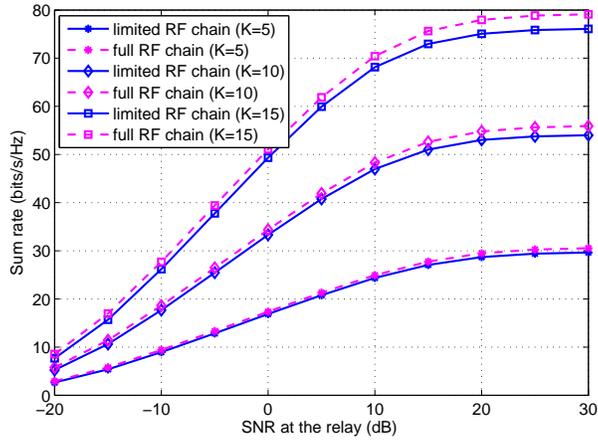}
\caption{Sum rate comparison with full RF chains where $P_r$ = 20dB, $N_d = N_r = 256$, and $\xi_k = 1$ for $k=1, \cdots, K$.}
\label{Fig:limited_full_comparison}
\end{figure}
\begin{figure}
\centering
\includegraphics[width=3.6in]{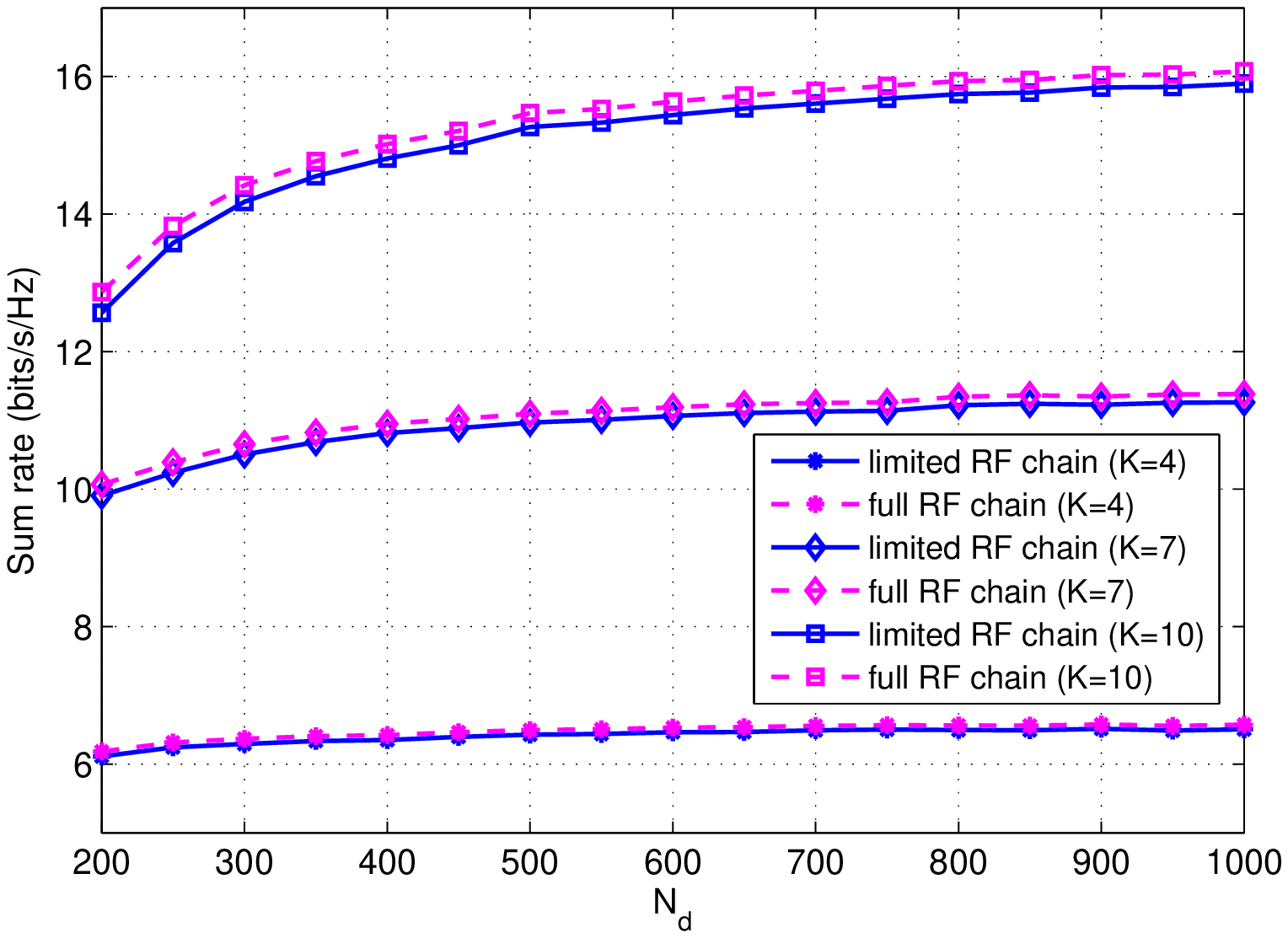}
\caption{Sum rate comparing with full RF chains where $P_u = E_u/N_d$, $P_r = E_u = 20$dB, $\delta=10$, and $\xi_k = 1$.}
\label{Fig:sum_rate_limited_full_RF_chains}
\end{figure}

Fig. \ref{Fig:sum_rate_comparison_simulation_approximation} shows the comparison of sum rate between simulation and analytical results. It shows that our analytical results can be a relatively accurate estimation. Besides, it can be obtained that the sum rate increases logarithmically with $N_d$ which is consistent with \eqref{eq:R_k2}.

Fig. \ref{Fig:limited_full_comparison} shows the achievable sum rate versus SNR at the relay. As the SNR at the relay becomes large, the sum rate, which is limited by the effect of noise amplification at the relay, converges to a constant. In addition, the hybrid detection of the limited-RF-chain case performs rather close to the ZF detection of the full-RF-chain case within a marginal gap.

Fig. \ref{Fig:sum_rate_limited_full_RF_chains} demonstrates the power scaling law of users derived in \eqref{eq:iii}. It implicates that user power can be scaled down by $1/N_d$ while keeping a nearly unchanged rate as $N_d\rightarrow\infty$. Comparing the sum rate of this system with that of the full-RF-chain case, we can see that the proposed hybrid detection causes negligible rate loss in contrast to fully digital detection.

\begin{figure}
\centering
\includegraphics[width=3.6in]{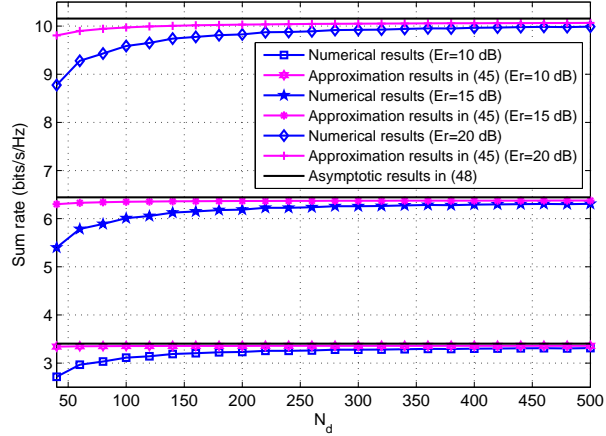}
\caption{Sum rate versus numbers of antennas with $P_u = E_u = 10$dB, $P_r = E_r/N_d$, $K = 5$, $\delta=1$, and $\xi_k = 1$.}
\label{Fig:power_scaling_law_relay}
\end{figure}

Fig. \ref{Fig:power_scaling_law_relay} demonstrates the power scaling law of relay derived in \eqref{eq:ii}. Our analytical results match well with the sum rate as the number of antennas goes larger. It implicates that the relay power can be scaled down by an factor of $1/N_d$ while keeping a nearly unchanged rate as $N_d\rightarrow\infty$. In addition, we can increase the value of $E_r$ when a larger sum rate is required.


\begin{figure}
\centering
\includegraphics[width=3.6in]{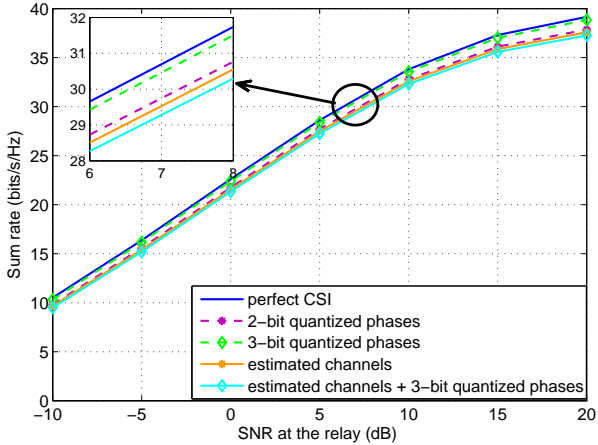}
\caption{Sum rate versus SNR under estimated cascade channels and quantized channel phase information, where $N_d=N_r=128$, $K=8$, and $P_r=20$~dB.}
\label{Fig:Sumrate_quantization_estimation}
\end{figure}

Fig. \ref{Fig:Sumrate_quantization_estimation} shows the impact of imperfect CSI including channel estimation error and quantized channel phase information.
Very recent studies including \cite{LPan2018}-\cite{YDing2018} proposed methods to realize channel estimation in massive MIMO with limited RF chains, even though there are still no tractable models for evaluating the channel estimation error for this case.
Therefore, we resort to Monte-Carlo simulations to evaluate the impact of imperfect CSI.
In Fig.~\ref{Fig:Sumrate_quantization_estimation}, it shows that the sum rate with 3-bit quantized phase shifters approaches that with ideal (continuous) phase shifters. When low-precision, e.g., 3-bit, quantized phase shifters and channel estimation error are considered, the rate loss is shown marginal compared to perfect CSI.

\begin{figure}
\centering
\includegraphics[width=3.7in]{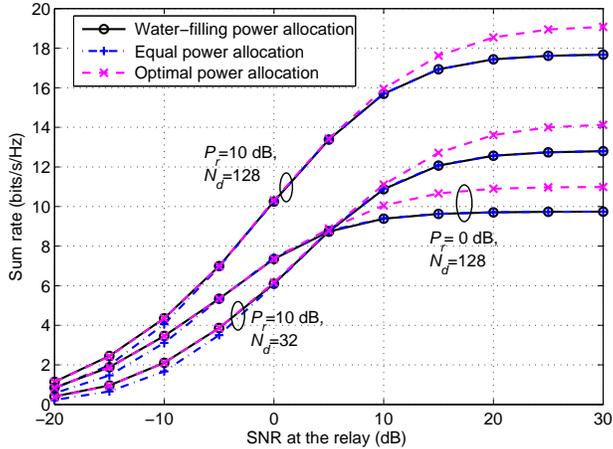}
\caption{ Sum rate comparison among water-filling PA, equal PA, and the proposed optimal PA with $K = 5$ and $N_r = N_d$.}
\label{Fig:power_allocation}
\end{figure}

\begin{figure}
\centering
\includegraphics[width=3.7in]{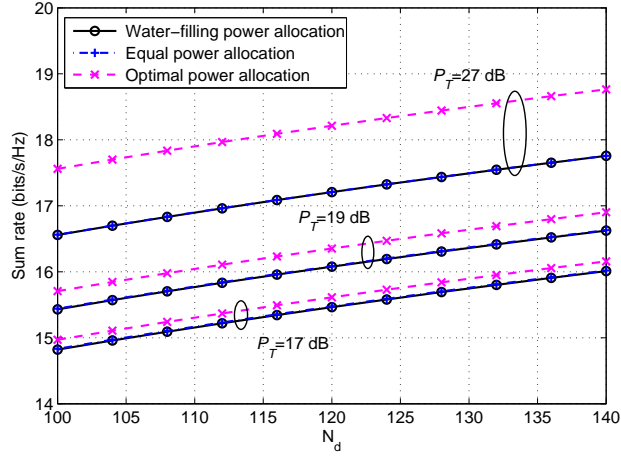}
\caption{Sum rate comparison among water-filling PA, equal PA, and the proposed optimal PA with $K = 5$, $P_r = 10$ dB, and $N_r = N_d$.}
\label{Fig:power_allocation_Nd}
\end{figure}

Fig. \ref{Fig:power_allocation} shows the comparison of the sum rate versus SNR with the proposed optimal PA and existing schemes including the water-filling PA and equal PA.
The optimal PA factors are obtained for various SNRs at the relay by applying three-dimensional search as expressed in Theorem \ref{theorem3}.
These schemes are compared under various system setups, i.e., \{$P_r=10$ dB, $N_d=128$\}, \{$P_r=0$ dB, $N_d=128$\}, and \{$P_r=10$ dB, $N_d=32$\}.
In each case, the proposed optimal PA achieves the best sum rate compared with both water-filling PA and equal PA.
At low SNRs, the optimal PA achieves almost the same sum rate as water-filling PA in each case.
This verifies our observation that the proposed optimal PA plays analogously as water-filling PA at low SNR regime in Section V.
While for high SNRs, the optimal PA noticeably outperforms the water-filling PA.
On the other hand, the gap between the optimal PA and equal PA is negligible for moderate SNR and becomes larger for higher or lower SNRs.
Fig. \ref{Fig:power_allocation_Nd} shows the comparison versus the number of antennas at the BS.
Three cases are tested with $P_T=17$~dB, $P_T=19$ dB, and $P_T=27$ dB.
As expected, the optimal PA outperforms both water-filling PA and equal PA in terms of sum rate.
Furthermore, the comparison shows that the sum rate gap between the optimal PA and the two existing schemes becomes more obvious with an increasing $P_T$.
This implies that the advantage provided by the optimal PA increases with the total transmit power of users.

\section{Conclusion}\label{sec:conclusion}
In this paper, we analyzed the massive MIMO relay system with limited RF chain at BS. The analytical expression of sum rate was obtained in closed form which is shown to be logarithmically proportional to $N_d$.
Power scaling laws were derived when the number of antennas tends to infinity. In general, users and relay can get an energy efficiency gain of $N^a_d$ and $N^{1-a}_d (0 < a < 1)$ simultaneously.
The optimal PA problem was solved by obtaining the closed form solutions with insightful engineering observations.

\appendices

\section{Useful Lemmas for the Proof of Theorem~\ref{theorem_rate}}

\begin{mylemma}\label{lemma2}
As $N_d\rightarrow \infty$, the matrix $\sqrt{N_rN_d}\mathbf{W}_d$ converges to the diagonal matrix as follows
\begin{equation}\label{eq:NrNdWd1}
\sqrt{N_rN_d}\mathbf{W}_d \!\xrightarrow{a.s.}\! {\rm diag}\left(\frac{2}{\sqrt{\pi \eta \xi_1}},\cdots, \frac{2}{\sqrt{\pi \eta \xi_k}},\cdots, \frac{2}{\sqrt{\pi \eta \xi_K}}\right).
\end{equation}
\end{mylemma}
\begin{proof}
Defining $\mathbf{H}_{eq}\triangleq \mathbf{GH}$ and $|h_{ij}|\triangleq \frac{1}{\sqrt{N_r}}\big|[\mathbf{GH}]_{ij}\big|$, we can write
\begin{equation}\label{eq:Heq2}
\frac{1}{\sqrt{N_r}}[\mathbf{H}_{eq}]_{ij} = \frac{1}{\sqrt{N_r}}[\mathbf{GH}]_{ij} = |h_{ij}|e^{-j\phi_{ji}}.
\end{equation}
Decompose the following matrix as
\begin{equation}\label{eq:WaHeq1}
\frac{1}{\sqrt{N_rN_d}}\mathbf{W}_a\mathbf{H}_{eq} = \mathbf{D}_w + \mathbf{A}_w,
\end{equation}
where $\mathbf{D}_w = {\rm diag}(d_{11},\cdots, d_{kk},\cdots, d_{KK})$ is a diagonal matrix and $[\mathbf{A}_w]_{kj} = d_{kj}~(k\neq j)$ while the diagonal entries of $\mathbf{A}_w$ are all zeros. Then we evaluate $d_{kk}$ and $d_{kj}$, separately.

For $d_{kk}$, we have
\begin{align}\label{eq:dkk1}
d_{kk} \overset{(a)}{=}\frac{1}{N_d}\sum_{i=1}^{N_d}|h_{ik}|
\overset{(b)}{\rightarrow}\mathbb{E}\{|h_{ik}|\},
\end{align}
where $(a)$ is obtained by using \eqref{eq:Fa} and \eqref{eq:Heq2}, and $(b)$ is derived by using the Law of Large Numbers because $|h_{ik}|~(\forall i\in\{1,2,\cdots,N_d\})$ are i.i.d. RVs
for any given $k$.

For $d_{kj}~(k\neq j)$, we have
\begin{align}\label{eq:dkj1}
d_{kj} \overset{(a)}{\rightarrow}\mathbb{E}\left\{\sqrt{\frac{N_d}{N_r}}\left[\mathbf{W}_a\right]_{ki}\left[\mathbf{H}_{eq}\right]_{ij}\right\}
\overset{(b)}{=}0,
\end{align}
where $(a)$ is derived by using the Law of Large Numbers because $\sqrt{\frac{N_d}{N_r}}\left[\mathbf{W}_a\right]_{ki}\left[\mathbf{H}_{eq}\right]_{ij}$ are i.i.d. RVs given $k,j$, and $(b)$ is obtained by using the independence of $\left[\mathbf{W}_a\right]_{ki}$ and $\left[\mathbf{H}_{eq}\right]_{ij}$ given $i$.

Combining \eqref{eq:dkk1} and \eqref{eq:dkj1}, we have \begin{align}
\frac{1}{\sqrt{N_rN_d}}\mathbf{W}_a\mathbf{H}_{eq} \xrightarrow{a.s.}
{\rm diag}(d_{11},\cdots, d_{kk},\cdots, d_{KK}).
\end{align}

However, it is difficult to calculate $d_{kk}$ due to the unknown distribution of $|h_{ik}|$ under a finite $N_d$. Fortunately, as $N_d\rightarrow\infty$, we can obtain an explicit closed-form value of $\mathbb{E}\{|h_{ik}|\}$.

Firstly, rewrite \eqref{eq:Heq2} as
\begin{align}
\frac{1}{\sqrt{N_r}}[\mathbf{H}_{eq}]_{ij}\!=\!\frac{1}{\sqrt{N_r}}\sum_{k=1}^{N_r}[\mathbf{G}]_{ik}[\mathbf{H}]_{kj}
\!\overset{(a)}{=}\!\frac{1}{\sqrt{N_r}}\sum_{k=1}^{N_r}(a_{ik}\!+\!jb_{ik})(c_{kj}\!+\!jd_{kj})
\!\overset{(b)}{=}\!\sum_{k=1}^{N_r}(e_{k}\!+\!jf_{k})
\!\overset{(c)}{\rightarrow}\! e'\!+\!jf',
\end{align}
where $(a)$ is obtained by writing $[\mathbf{G}]_{ik}\triangleq a_{ik}+jb_{ik}$ and $[\mathbf{H}]_{kj}\triangleq c_{kj}+jd_{kj}$, $(b)$ is obtained by rewriting $e_{k}\triangleq \frac{1}{\sqrt{N_r}}(a_{ik}c_{kj}-b_{ik}d_{kj})$ and $f_{k}\triangleq \frac{1}{\sqrt{N_r}}(b_{ik}c_{kj}+a_{ik}d_{kj})$, and  $(c)$ is obtained by defining $e'\triangleq \lim\limits_{N_r\rightarrow\infty}\sum_{k=1}^{N_r}e_{k}$ and $f'\triangleq \lim\limits_{N_r\rightarrow\infty}\sum_{k=1}^{N_r}f_{k}$. Using the distribution of $[\mathbf{G}]_{ik}$ and $[\mathbf{H}]_{kj}$ and the independence of $a_{ik},b_{ik},c_{kj},d_{kj}$, we can easily know that $e_k$ and $f_k$ both have mean 0 and variance $\frac{\eta\xi_j}{2 N_r}$. Applying the Central Limit Theorem, we know that as $N_r\rightarrow\infty$, the RV $e$ and $f$ both tend to follow a normal distribution with mean 0 and variable $\frac{1}{2}\eta\xi_j$, that is, $e'\sim\mathcal{N}(0, \frac{1}{2}\eta\xi_j)$ and $f'\sim\mathcal{N}(0, \frac{1}{2}\eta\xi_j)$. Define a random vector $\mathbf{m}_{ij}$ as follows:
\begin{equation}
\mathbf{m}_{ij}=\sum_{k=1}^{N_r}[e_{k}~f_{k}]^T.
\end{equation}
According to the Multidimensional Central Limit Theorem given by \cite{hahn1981the}, as $N_r\rightarrow\infty$, $\mathbf{m}_{ij}$ converges to a 2-dimensional normal distribution, that is, the RVs $e'$ and $f'$ are jointly Gaussian. Since the covariance of $e'$ and $f'$ is easily calculated to be zero, $e'$ and $f'$ are uncorrelated. Having proved that $e'$ and $f'$ are jointly Gaussian, we therefore have the independence of $e'$ and $f'$. Consequently, as $N_r\rightarrow\infty$, $\frac{1}{\sqrt{N_r}}[\mathbf{H}_{eq}]_{ik}$ is asymptotically a complex normal RV, i.e.,
\begin{equation}\label{eq:NrGH}
\frac{1}{\sqrt{N_r}}[\mathbf{GH}]_{ij}\sim\mathcal{CN}(0, \eta\xi_j).
\end{equation}
From \eqref{eq:NrGH}, $|h_{ik}|$ follows Rayleigh distribution with mean $\sqrt{\pi\eta\xi_k}/2$ for $N_d\rightarrow\infty$. Therefore,
\begin{align}\label{eq:WaHeqa.s.1}
\frac{\mathbf{W}_a\mathbf{H}_{eq}}{\sqrt{N_rN_d}} \xrightarrow{a.s.}
{\rm diag}\left(\frac{\sqrt{\pi\eta\xi_1}}{2}, \cdots, \frac{\sqrt{\pi\eta\xi_k}}{2}, \cdots, \frac{\sqrt{\pi\eta\xi_K}}{2}\right).
\end{align}
By applying the Continuous Mapping Theorem \cite{convergence} and using \eqref{eq:WaHeqa.s.1} and \eqref{eq:dig_detector}, we can obtain \eqref{eq:NrNdWd1}.
\end{proof}

\begin{mylemma}\label{lemma3}
As $N_d\rightarrow \infty$, we have
\begin{equation}
N_d\mathbb{E}\left\{\left[\mathbf{W}_d \mathbf{W}_a \mathbf{G} \mathbf{G}^H \mathbf{W}_a^H \mathbf{W}_d^H\right]_{kk}\right\} \xrightarrow{a.s.} \frac{\pi \delta + 4}{\pi \xi_k}.
\end{equation}
\end{mylemma}

\begin{proof}
Based on Proposition~\ref{mypro1} and Lemma \ref{lemma2}, we have
\begin{align}
\nonumber
N_d\mathbb{E}\left\{\left[\mathbf{W}_d \mathbf{W}_a \mathbf{G} \mathbf{G}^H \mathbf{W}_a^H \mathbf{W}_d^H\right]_{kk}\right\}
\overset{(a)}{\rightarrow}& \frac{1}{N_r}\mathbb{E}\left\{\frac{2}{\sqrt{\pi \eta \xi_k}} \left[\mathbf{W}_a \mathbf{G} \mathbf{G}^H \mathbf{W}_a^H\right]_{kk} \frac{2}{\sqrt{\pi \eta \xi_k}}\right\}\\\nonumber
\overset{(b)}{=} & \frac{\pi N_d + 4N_r - \pi}{\pi \xi_k N_r}\\
\overset{(c)}{\rightarrow} & \frac{\pi \delta + 4}{\pi \xi_k},
\end{align}
where $(a)$ follows from Lemma \ref{lemma2} and the Continuous Mapping Theorem \cite{convergence}, $(b)$ applies Proposition~\ref{mypro1}, and $(c)$ is obtained by using $N_d/N_r \triangleq\delta$ and $N_d\rightarrow\infty$.
\end{proof}

\begin{mylemma}\label{lemma4}
As $N_d, N_r\rightarrow \infty$, we have
\begin{equation}
N_d\mathbb{E}\left\{\frac{\left[\mathbf{W}_d \mathbf{W}_a \mathbf{W}_a^H \mathbf{W}_d^H\right]_{kk}}{\alpha^2}\right\} \xrightarrow{a.s.} \frac{4P_u \sum_{i=1}^{K}\xi_i + 4\sigma_r^2}{\pi \eta \xi_k P_r}.
\end{equation}
\end{mylemma}

\begin{proof}
Based on Proposition~\ref{mypro2} and Lemma \ref{lemma2}, we have
\begin{align}
\nonumber
N_d\mathbb{E}\left\{\frac{\left[\mathbf{W}_d \mathbf{W}_a \mathbf{W}_a^H \mathbf{W}_d^H\right]_{kk}}{\alpha^2}\right\}
\overset{(a)}{=}&\frac{N_d}{P_r}\mathbb{E}\left\{\left(P_u \text{Tr}(\mathbf{H}^H\mathbf{H}) + \sigma_r^2 N_r\right)\left[\mathbf{W}_d \mathbf{W}_a \mathbf{W}_a^H \mathbf{W}_d^H\right]_{kk}\right\}\\\nonumber
\overset{(b)}{\rightarrow} & \frac{N_rN_d \left(P_u \sum_{i=1}^{K}\xi_i + \sigma_r^2 \right)}{P_r}\mathbb{E}\left\{\left[\mathbf{W}_d \mathbf{W}_a \mathbf{W}_a^H \mathbf{W}_d^H\right]_{kk}\right\}\\\nonumber
\overset{(c)}{\rightarrow} & \frac{4P_u \sum_{i=1}^{K}\xi_i + 4\sigma_r^2}{\pi \eta \xi_k P_r}\mathbb{E}\left\{\left[\mathbf{W}_a \mathbf{W}_a^H \right]_{kk}\right\}\\
\overset{(d)}{\rightarrow} & \frac{4P_u \sum_{i=1}^{K}\xi_i + 4\sigma_r^2}{\pi \eta \xi_k P_r},
\end{align}
where $(a)$ is obtained by using \eqref{eq:alpha} and $(b)$ is derived by applying the Law of Large Numbers that $\frac{1}{N_r}$\text{Tr}$(\mathbf{H}^H\mathbf{H}) \xrightarrow{a.s.} \sum_{i=1}^{K}\xi_i$ when $N_r\rightarrow\infty$. Employing Lemma \ref{lemma2} and Proposition~\ref{mypro2} respectively, followed by utilizing the Continuous Mapping Theorem \cite{convergence}, we obtain $(c)$ and $(d)$.
\end{proof}

\section{Useful Results for the Proof of Theorem~\ref{theorem_rate_Q}}

For quantized phase information, some similar results as in Lemma~\ref{lemma3}-\ref{lemma4} can be obtained for the proof of Theorem~\ref{theorem_rate_Q}, after analogous derivations for perfect phase information.
For notational brevity, a brief description for the proofs of these useful results are given in this appendix, where we use $\hat{\mathbf{W}}_a$ and $\hat{\mathbf{W}}_d$ instead of $\mathbf{W}_a$ and $\mathbf{W}_d$ respectively.

Firstly, a similar conclusion can be obtained as in Theorem~\ref{mytheorem1}.
The major difference is the derivation of $\mathbb{E}\left\{\mathbf{g}_i^H\mathbf{g}_j e^{j(\phi_1-\phi_2)}\right\}$ for $i\neq j$, which originates from a different result from \eqref{eq:gi1}. In the case of quantized phase information, \eqref{eq:gi1} changes to
\begin{align}
\mathbb{E}_{g_{i1},\epsilon_{i1}}\left\{g_{i1} e^{-j(\varphi_{i1}+\epsilon_{i1})}\right\}
=\mathbb{E}\left\{|g_{i1}|\right\}\mathbb{E}\left\{e^{-j\epsilon_{i1}}\right\} \overset{(a)}{=}\frac{\sqrt{\pi\eta}}{2}{\rm sinc}(\frac{\pi}{2^{b}}),
\label{eq:Q}
\end{align}
where $\epsilon_{i1}$ is the quantized error of $\varphi_{i1}$ and $(a)$ uses the Rayleigh distribution of $|g_{i1}|$ and the uniform distribution of $\epsilon_{i1}$, i.e., $\epsilon_{i1}\sim U[-\frac{\pi}{2^{b}},\frac{\pi}{2^{b}})$.
By replacing \eqref{eq:Q} with \eqref{eq:gi1} in the proof of Theorem~\ref{mytheorem1}, a similar result is obtained as follows
\begin{equation}\label{eq:E gigj_Q}
\mathbb{E}\left\{\mathbf{g}_i^H\mathbf{g}_j e^{j(\hat{\phi}_1-\hat{\phi}_2)}\right\} = \begin{cases}
\frac{\pi\eta}{4}{\rm sinc}^2(\frac{\pi}{2^{b}})  &i\neq j\\
\eta N_r           &i=j,
\end{cases}
\end{equation}
where $\hat{\phi}_1$ and $\hat{\phi}_2$ are quantized phases.

Then, using \eqref{eq:E gigj_Q} and after analogous proofs, similar results as in Propositions~\ref{mypro1}-\ref{mypro2} and Lemma~\ref{lemma2} are obtained as
\begin{equation}\label{eq:appro_fggf_Q}
\mathbb{E}\left\{\left[\hat{\mathbf{W}}_a\mathbf{G}\mathbf{G}^H\hat{\mathbf{W}}_a^H\right]_{kk} \right\} \!=\! \eta \left(\frac{\pi(N_d\!-\!1)}{4}{\rm sinc}^2\Big(\frac{\pi}{2^b}\Big)\!+\! N_r\right),
\end{equation}
\begin{equation}\label{eq:mypro2_Q}
\hat{\mathbf{W}}_a\hat{\mathbf{W}}_a^H\xrightarrow{a.s.}\mathbf{I}_K,
\end{equation}
\begin{align}
\sqrt{N_rN_d}\hat{\mathbf{W}}_d \xrightarrow{a.s.}
 \frac{1}{{\rm sinc}\big(\frac{\pi}{2^b}\big)}{\rm diag}\left(\frac{2}{\sqrt{\pi \eta \xi_1}},\cdots, \frac{2}{\sqrt{\pi \eta \xi_k}},\cdots, \frac{2}{\sqrt{\pi \eta \xi_K}}\right).
\label{eq:lemma1_Q}
\end{align}

Finally, similar results as in Lemmas~\ref{lemma3}-\ref{lemma4} can be further obtained by using the updated conclusions in \eqref{eq:appro_fggf_Q}-\eqref{eq:lemma1_Q}.
Specifically, for $N_d, N_r\rightarrow \infty$, we have
\begin{align}
N_d\mathbb{E}\left\{\left[\hat{\mathbf{W}}_d \hat{\mathbf{W}}_a \mathbf{G} \mathbf{G}^H \hat{\mathbf{W}}_a^H \hat{\mathbf{W}}_d^H\right]_{kk}\right\}
\xrightarrow{a.s.} \frac{\pi N_d {\rm sinc}^2\big(\frac{\pi}{2^b}\big) + 4N_r}{\pi \xi_k N_r {\rm sinc}^2\big(\frac{\pi}{2^b}\big)},
\label{eq:lemma2_Q}
\end{align}
\begin{equation}
N_d\mathbb{E}\!\left\{\!\!\frac{\left[\!\hat{\mathbf{W}}_d \hat{\mathbf{W}}_a \hat{\mathbf{W}}_a^H \hat{\mathbf{W}}_d^H\!\right]_{kk}}{\alpha^2}\!\!\right\} \!\!\xrightarrow{a.s.} \!\!\frac{4P_u\! \sum_{i=1}^{K}\xi_i \!+\! 4\sigma_r^2}{\pi \eta \xi_k P_r{\rm sinc}^2\big(\frac{\pi}{2^b}\big)}\!.
\label{eq:lemma3_Q}
\end{equation}

\end{document}